%% file: manuscript.tex
\documentclass[12pt]{article}

\input{header}

\title{\sc Targeted Advertising in Elections\thanks{
I thank Avidit Acharya, Sandeep Baliga, Arjada Bardhi, Renee Bowen, Steve Callander, Georgy Egorov, Catherine Hafer, Federica Izzo, Aleksandr Levkun, Navin Kartik, Joel Sobel, Francesco Trebbi, Richard Van Weelden, Kun Zhang and the audience members at 
Bilkent,
CEU,
CMU/Pittsburgh,
Columbia,
Essex,
Glasgow,
Linz, 
LSE,
NBER Political Economy,
NYU,
Princeton,
SAET,
SITE,
Stanford GSB,
UC Riverside,
UC San Diego,
UNC Chapel Hill,
Warwick. All errors are my own.}}
\author{Maria Titova\footnote{Department of Economics, Vanderbilt University. E-mail: motitova@gmail.com.}}

\date{\monthyeardate\today
}

\begin{document}

\maketitle

\vspace{-0.5cm}

\begin{abstract}

    How does targeted advertising influence electoral outcomes? This paper presents a one-dimensional spatial model of voting in which a privately informed challenger persuades voters to support him over the status quo. I show that targeted advertising enables the challenger to persuade voters with opposing preferences and swing elections decided by such voters; under simple majority, the challenger can defeat the status quo even when it is located at the median voter's bliss point. Ex-ante commitment power is unnecessary---the challenger succeeds by strategically revealing different pieces of verifiable information to different voters. Publicizing all political ads would mitigate the negative effects of targeted advertising and help voters collectively make the right choice.

\end{abstract}

{\sc keywords}: Persuasion, verifiable information,
targeted advertising, elections

    {\sc jel classification}: D72, D82, D83

\onehalfspacing

\clearpage

\section{Introduction}

Targeted advertising, broadly defined as private messaging aimed at specific groups of voters, played a key role in many successful electoral campaigns. In 1960, John F. Kennedy's campaign distributed two million copies of ``the blue bomb''---a pamphlet advertising his support of civil rights---to African American churches across the U.S. Decades later, George W. Bush's 2004 reelection campaign used direct mail to communicate his opposition to gay marriage and support for ``traditional family values'' to evangelical Christian households. More recently, the 2016 Brexit referendum and the Trump presidential campaign both employed the services of Cambridge Analytica, a data mining firm, to design and distribute thousands of targeted ads to diverse audiences.
Although these examples suggest broad awareness of a correlation between targeted advertising and electoral success, the precise mechanisms by which tailoring messages to different voters influences election outcomes remain poorly understood.

This paper proposes a simple theoretical model that fills this gap.
The model is grounded in three stylized facts about electoral campaigns.
First, voters have incomplete information and update their beliefs in response to campaign messages (\citenop{kendall_how_2015}; \citenop{spenkuch_political_2018}; \citenop{le_pennec_how_2023}).
Second, politicians use the strategy of ambiguity to avoid making precise statements about their positions on issues (\citenop{Page1978}; \citenop{druckman_campaign_2009}; \citenop{fowler_political_2021}).
Third, politicians tailor messages to specific groups of voters (\citenop{Hillygus2014}).
I identify a novel mechanism by which targeted advertising changes electoral outcomes.
In particular, I show that privately revealing different pieces of information to different voters allows politicians to persuade voters with diametrically opposing preferences and win otherwise unwinnable elections, in which such voters are pivotal.

My model has two components: an advertising campaign followed by an election.\footnote{The analysis is not limited to political campaigns but applies more broadly to environments wherein an informed agent seeks to persuade a group of decision makers with different preferences, for example, in lobbying, corporate governance, hiring and other organizational contexts.}
In the election, a unit mass of voters chooses between two options, the challenger and the status quo.
Voters care about the candidates' {\it policy outcomes}, which represent proposed policies, their implementation, or welfare consequences.
The challenger's policy outcome $x \in [-1,1]$ is initially unknown to the voters, while the status quo policy outcome is commonly known and normalized to zero.
Voters have single-peaked and single-crossing preferences, with the leading example being quadratic loss.\footnote{In this setting, single crossing requires that the voters' utility difference between the status quo and {\it any} lottery over challenger's policy outcomes changes sign at most once if we ordered voters from left to right.} The goal of the office-motivated challenger is to convince a decisive coalition of voters to approve his proposal. I model the challenger's advertising campaign as a game of persuasion with verifiable information (\citenop{grossmanDisclosureLawsTakeover1980}, \citenop{Milgrom1981} and \citenop{Grossman1981}). That is, I assume that the challenger privately knows his policy outcome $x$ and can send any subset of $[-1,1]$ that contains $x$. Conceptually, this communication protocol allows the challenger to lie by omission but not commission: a message $[-0.5,0]$ informs a voter that the challenger's policy outcome is moderately left, but is only partially informative because $x$ could be anywhere between $-0.5$ and $0$. Communication with verifiable information is a reasonable middle ground between the possibilities identified by \cite{Persson2002}, who famously wrote (p. 483), ``It is thus somewhat schizophrenic to study either extreme: where promises have no meaning or where they are all that matter.'' I consider two versions of the game: public advertising  and targeted advertising. The former models a public advertising campaign in which the challenger sends the same message to all voters. The latter models a targeted advertising campaign in which the challenger knows the voters' bliss points and sends private messages to different groups of voters.

The main contribution of the paper is showing that targeted advertising allows the challenger to win elections that are unwinnable with public advertising.\footnote{I say that an election is unwinnable with public or targeted advertising if the challenger's ex-ante odds of winning are zero in every (perfect Bayesian) equilibrium of the corresponding game. An election is winnable otherwise.}
My first two results (Theorems \ref{thm1:PA} and \ref{thm2:TA}) characterize elections (described by a set of voters' bliss points and a set of decisive coalitions) that are unwinnable for the challenger with public and targeted advertising, respectively.
\cref{thm1:PA} states that an election is unwinnable for the challenger with public advertising if and only if there is no decisive coalition of left or right voters.
The intuition is simple: if there is a decisive coalition of left (right) voters, then the challenger can use a fully revealing strategy and win when his policy outcome is left (right) of the status quo.
Otherwise, all decisive coalitions include status quo voters (whose bliss point is the status quo) or left and right voters (who have diametrically opposing preferences, so the status quo is already the best compromise).
\cref{thm2:TA} states that an election is unwinnable for the challenger with targeted advertising if and only if every decisive coalition includes a status quo voter.
In particular, targeted advertising makes it possible for the challenger to convince voters on the opposite sides of the status quo with a positive probability by telling them different things.
Under simple majority, Theorems \ref{thm1:PA} and \ref{thm2:TA} classify elections as follows:
\begin{itemize}
    \item[I.] If the median voter's bliss point is left or right of the status quo, then the election is winnable with public advertising.
    \item[II.] If the median voter's bliss point is at the status quo, but status quo voters do not form a majority, then the election is unwinnable with public advertising but winnable with targeted advertising.
    \item[III.] If the median voter's bliss point is at the status quo and status quo voters do form a majority, then the election is unwinnable with public or targeted advertising.
\end{itemize}

The second part of the paper focuses on the optimal targeted advertising strategy that maximizes the challenger's (ex-ante) odds of winning elections that are unwinnable with public advertising.
I make a further simplification that all left voters have the same bliss point $L<0$ and all right voters have the same bliss point $R>0$.
Below I use a motivating example in which $L = -0.2$, $R = 0.4$, the voters' preferences are quadratic, the minimal decisive coalition includes left and right voters, and $x \sim U[-1,1]$, to illustrate the following two results:
\cref{prop:baseline-TA} identifies the optimal targeted advertising strategy, while
\cref{prop:comp_stats} describes the comparative statics as the right voters become more extreme/the electorate becomes more polarized.

Consider the following strategy of the challenger: to the left voters, he reveals whether his policy outcome is in the set $[-0.4,0.2]$, or not.\footnote{That is, the challenger sends message $[-0.4,0.2]$ if $x \in [-0.4,0.2]$ and message $[-1,0.4)\cup(0.2,1]$ if $x \notin [-0.4,0.2]$ to all left voters} To the right voters, he reveals whether his policy outcome is in $[-0.4,0.8]$, or not. When a left voter receives message $[-0.4,0.2]$, she learns that the challenger's policy outcome could be anywhere  in $[-0.4,0.2]$, which is just enough information to convince her to approve.\footnote{If the prior is uniform on $[-1,1]$, then a left voter's posterior belief (calculated via Bayes rule given the challenger's strategy) after message $[-0.4,0.2]$ is uniform on $[-0.4,0.2]$. Her expected utility is $\int_{-0.4}^{0.2} -\frac{(x+0.2)^2}{0.6} dx = -0.04$ if she approves and $-0.2^2 = -0.04$ (since the status quo policy outcome is normalized to zero) if she rejects.} By similar reasoning, a right voter is convinced after message $[-0.4,0.8]$. This strategy leads to the following electoral outcome: the left voters approve if and only if $x \in [-0.4,0.2]$ and the right voters approve if and only if $x \in [-0.4,0.8]$. Given that left and right voters form a decisive coalition, the challenger wins the election if and only if his policy outcome is between $-0.4$ and $0.2$. His ex-ante odds of winning are 30\% -- a massive improvement over his odds of winning without targeted advertising, which are 0\%. \cref{prop:baseline-TA} confirms that the described electoral outcome is an equilibrium outcome with the highest odds of the challenger winning across all equilibria of the targeted advertising game.

\begin{figure}[ht!]
    \centering
    \input{figures/fig-intro.tex}
    \caption{Targeted messages that convince left voters (in blue) and right voters (in red). The challenger wins the election whenever his policy outcome lies in the intersection of the convincing messages (in black).}
    \label{fig:intro}
\end{figure}
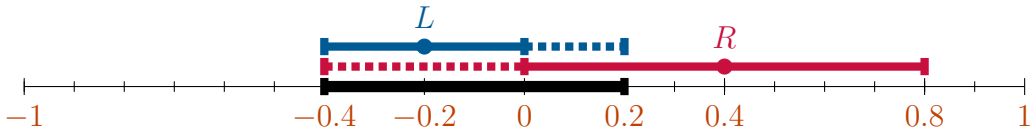

To see how this challenger-preferred equilibrium outcome changes as right voters become more extreme, suppose the right voters' bliss point increases from $R = 0.4$ to $R' = 0.5$. Following the same logic as above, we find that the convincing messages are $[-0.4,0.2]$ for left voters and $[-0.5,1]$ for right voters. The challenger wins when his policy outcome is between $-0.4$ and $0.2$, exactly as before. However, the challenger's equilibrium odds of winning may be even higher. Specifically, observe that when the challenger's policy outcome is between $-0.5$ and $-0.4$, the strategy described above convinces right but not left voters. However, left voters actually prefer policy outcomes in $[-0.5,-0.4]$ to those in $[0.1,0.2]$ as the former are closer to their bliss point. Hence, we can recalculate the message that convinces left voters (making them indifferent between approval and rejection), forcing it to start at $-0.5$. That message is $[-0.5,0.179]$. \cref{fig:intro-2} illustrates the electoral outcome after right voters become more extreme.

\begin{figure}[ht!]
    \centering
    \input{figures/fig-intro-2}
    \caption{More extreme right voters are persuadable by policy outcomes further to the left. As a result, the set of the challenger's winning policy outcomes (in black) is larger and shifts to the left.}
    \label{fig:intro-2}
\end{figure}
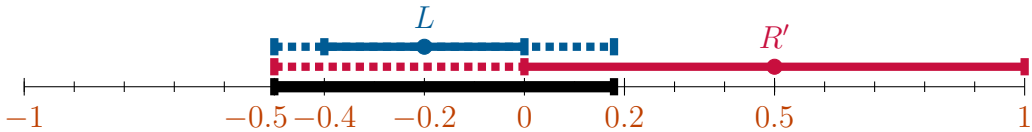

In the new equilibrium, the set of winning policy outcomes is $[-0.5,0.179]$ and the challenger's odds of winning are 33.96\%.
\cref{prop:comp_stats} confirms that when right voters become more extreme, the set of winning policy outcomes shifts in the opposite direction, to the left; also, the challenger's odds of winning increase. Intuitively, when right voters become more extreme, their dissatisfaction with the status quo grows, which makes them persuadable by wider ranges of policy outcomes.

My findings suggest a novel explanation for why politicians use the strategy of ambiguity: advertising different ranges of policy outcomes to different voters allows politicians to persuade voters with diametrically opposing preferences without lying (by commission) to any of them. Previous explanations for why politicians use the strategy of ambiguity include voters’ risk-seeking behavior (\citenop{Shepsle1972}),  candidates’ preference for ambiguity (\citenop{Aragones2000}), subsequent elections (\citenop{Meirowitz2005}, \citenop{Alesina2008}), resolution of uncertainty after an election (\citenop{Kartik2017}). Two previous papers find that ambiguity enables politicians to persuade voters with opposing preferences: in \cite{Callander2008}, voters have context-dependent preferences, and in \cite{tolvanenPoliticalAmbiguityAntiMedian2024}, the voters’ preferences are correlated with the state of the world. I reach a similar conclusion in a setting where voters have standard single-peaked and single-crossing preferences.

This paper builds on the literature comparing public and private communication in both electoral environments and broader sender-receiver settings. 
When messages are verifiable (like in this paper) and candidates are symmetric, unraveling occurs: competing politicians voluntarily disclose all information, whether advertising is public or private (\citenop{Janssen2017}; \citenop{schipper2019}).
However, I show that when candidates are asymmetric---specifically, when only the challenger advertises privately---there are equilibria without unraveling, and the challenger generally prefers private to public communication with verifiable information. 
In cheap-talk models, by contrast, senders often favor public communication because it limits the number of possible deviations in each state of the world (\citenop{Farrell1989}, \citenop{Koessler2008}, \citenop{Goltsman2011}, \citenop{Bar-Isaac2014}).\footnote{An exception is \cite{Schnakenberg2015} where a cheap-talk sender prefers private communication when the policy space is multidimensional, because then he can make statements about different dimensions of the policy to different voters. In contrast, my sender prefers private communication even when the policy space is one-dimensional.}
Consequently, targeted advertising cannot swing unwinnable elections if ads consist only of cheap talk. My analysis thus highlights that persuading voters with opposing preferences requires providing selective evidence or easily verifiable facts. 
Finally, in the information-design literature, the sender has ex-ante commitment power and generally prefers private communication (\citenop{Arieli2019}, \citenop{Chan2019}, \citenop{Heese2019}). 
That said, public and private information design yield the same outcomes if receivers possess private information and choose between two actions (\citenop{kolotilinPersuasionPrivatelyInformed2017}),
if receivers compete (\citenop{asseyerInformationDesignCompeting2025}),
and in elections with unanimity voting (\citenop{Bardhi2018}). 
Since the challenger-preferred equilibrium strategy described in \cref{prop:baseline-TA} is also an information design solution, my analysis identifies an electoral environment in which the sender strictly prefers private to public information design.

My analysis suggests that targeted advertising is bad for democracy because it elects politicians who are guaranteed to lose when voters possess the same information.
For example, under simple majority, targeted advertising allows the challenger to beat the status quo located at the median voter's bliss point, which, according to various versions of the median voter theorem, is unbeatable.
The most effective policy to make targeted advertising obsolete is to publicize all ads transmitted during electoral campaigns.
While voters may still make mistakes due to incomplete (but public) information, having a common belief would be sufficient for them to collectively make the right choice.

\section{Model}\label{section:model}

There is a challenger (he/him) and a unit mass of voters (she/her). The space of policy outcomes is $X:=[-1,1]$.
Let $V := \{ v_1,\ldots,v_n \}$, where $-1 \leq v_1 < \ldots < v_n \leq 1$, denote the ordered set of voters' bliss points; I will refer to $V$ as the {\it electorate}.\footnote{
    To simplify exposition, I assume that $V$ is finite. The results extend to infinite electorates under mild regularity conditions.
}
I refer to a voter with bliss point $v \in V$ as ``voter $v$'' when there is no possibility of confusion. The game proceeds as follows.

\begin{enumerate}
    \item The challenger learns his policy outcome $x \in X$, drawn from a common prior distribution $\mu_0 \in \Delta X$ that has full support and no atoms.\footnote{For a topological space $Y$, I let $\Delta Y$ denote the set of all Borel probability measures over $Y$, endowed with the weak* topology. I say that $\gamma \in \Delta Y$ is degenerate if $\gamma(y) = 1$ for some $y \in Y$, denoted by $\gamma = \delta_y$, and non-degenerate otherwise. Unless otherwise specified, all subsets $W \subseteq X$ are assumed to be Borel. For any $W \subseteq X$, denote its complement by $W^c := X \smallsetminus W$ and let $\floor{W} := \inf W$ and $\ceil{W} := \sup W$. For any $W \subseteq X$ with $\mu_0(W) > 0$, let $\mu_0(\cdot \sep W)$ denote the conditional prior measure, defined by $\mu_{0}(B \sep W) := \mu_{0}(B \cap W)/\mu_{0}(W)$ for all $B \subseteq X$.
    }

    \item The challenger sends messages to voters.
          Each message $m$ is a Borel subset of $X$ (a statement about his policy outcome) that contains a grain of truth, $x \in m$. This communication protocol, introduced by \cite{grossmanDisclosureLawsTakeover1980}, \cite{Milgrom1981}, and \cite{Grossman1981}, allows the challenger to {\it lie by omission} and send messages that contain policy outcomes other than $x$. However, it does not allow the challenger to {\it lie by commission} and send messages that do not include $x$. I consider two versions of the game:
          \begin{itemize}
              \item {\bf targeted advertising}: the challenger chooses a collection of private messages $( m_v )_{v \in V}$, and voters with bliss point $v \in V$ observe message $m_v$ only;
              \item {\bf public advertising}: the challenger chooses a public message $m$ that is the same for all $v \in V$.
          \end{itemize}

    \item Each voter decides whether to {\it approve} the challenger's policy outcome or {\it reject} it in favor of the status quo. I normalize the status quo policy outcome to $0$.

    \item Payoffs are realized. The outcome of the vote is determined by a voting rule that, for any profile of voters' approve/reject decisions, specifies whether the challenger wins (and receives a payoff of $1$) or loses (and gets $0$). 
    I say that a coalition of voters $D \subseteq V$ is {\it decisive} if the challenger {\it (i)} wins whenever all voters in $D$ approve and {\it (ii)} loses whenever all voters in $D$ reject.\footnote{I am grateful to the editor for suggesting this definition.} I let $\mathcal{D}$ denote the set of decisive coalitions. By construction, $\mathcal{D}$ is monotonic ($D \in \mathcal{D}$ and $D \subset D'$ imply $D' \in \mathcal{D}$) and proper ($D \in \mathcal{D}$ implies $V \smallsetminus D \notin \mathcal{D}$).\footnote{Monotonicity follows directly from the definition of a decisive coalition. 
    Properness holds because if both $D$ and $V \smallsetminus D$ were decisive, then the collective choice would be ill-defined when all voters in $D$ approve and all voters in $V \smallsetminus D$ reject.} I assume that $\mathcal{D}$ is nonempty (equivalently, $V \in \mathcal{D}$) and 
    refer to the pair $(V,\mathcal{D})$ as the {\it election}. 

    Voters are expressive, meaning that the payoff $u_v$ of voter $v \in V$ depends on the policy outcome that she votes for (which is $x$ if she approves, and $0$ if she rejects).\footnote{Expressive voters derive utility from expressing support, whether based on ethics, identity, or ideology, for one of the candidates, independent of any effect of the voting act on the electoral outcome. See \cite{Brennan1993} and \cite{Hamlin2011} for the theory of expressive voting behavior and \cite{Felsenthal1985}, \cite{Kan2001}, \cite{Artabe2014} for empirical evidence of it.} I describe voter $v$'s preferences using her {\it net payoff from approval}, $\alpha_v(x) := u_v(x) - u_v(0)$, so that $v$ weakly prefers to approve $x \in X$ whenever $\alpha_v(x) \geq 0$. I let voter $v$'s \textit{approval set} be the set $A_v := \{ x \in X \sep \alpha_v(x) \geq 0 \}$ of policy outcomes that she prefers to approve under complete information.
\end{enumerate}

I assume that $\alpha_v(x)$ is continuous, measurable, bounded, and satisfies two properties standard in spatial voting models:
\begin{assumption}\label{ass1:single-peaked}
    Voters' preferences are \ul{single-peaked} (strictly quasiconcave): for each voter $v \in V$, her net payoff from approval $\alpha_v(x)$ is strictly increasing on $[-1, v]$ and strictly decreasing on $[v, 1]$.
\end{assumption}
One relevant consequence of \cref{ass1:single-peaked} is that voter $v$'s upper contour sets are convex. In particular, the approval set $A_v$ is an interval $[ \floor{A_v}, 0 ]$ for {\it left voters} with $v < 0$, a point $\{ 0 \}$ for {\it status quo voters} with $v = 0$, and an interval $[0,\ceil{A_v}]$ for {\it right voters} with $v > 0$.

\begin{assumption}\label{ass2:single-crossing}
    Voters' preferences are \ul{single-crossing}: for each belief $\mu \in \Delta X \smallsetminus \delta_0$,
    \begin{equation}
        \begin{split}
            \text{for all } v, w \in V \text{ such that } v w < 0, \quad
            \mathbbm{E}_\mu[\alpha_v(x)] \geq 0 \implies \mathbbm{E}_\mu[\alpha_w(x)] < 0,
        \end{split}
        \tag{SC1}\label{SC1:incompatible}
    \end{equation}
    and
    \begin{equation}
        \begin{split}
             & \text{for all } v, w \in V \text{ such that } w < v < 0 \text{ or } 0 < v < w,      \\
             & \mathbbm{E}_\mu [\alpha_v(x)] \geq 0 \implies \mathbbm{E}_\mu [\alpha_w(x)] \geq 0.
        \end{split}
        \tag{SC2}\label{SC2:more-extreme}
    \end{equation}~\\[-2\baselineskip]
\end{assumption}

I refer to \cref{ass2:single-crossing} as a single-crossing property because it states that the sign of $\mathbbm{E}_\mu [\alpha_v(x)]$ is monotone in $v$ for all $\mu$.\footnote{Note that \cref{ass2:single-crossing} is
    implied by single-crossing expectational differences, which requires that ``an agent’s utility difference between any pair of lotteries ... is single crossing in the agent’s preference parameter or type'' (\citenop{Kartik2023}, p. 2982). Since the status quo is fixed and known, \cref{ass2:single-crossing} instead requires that the voters' utility difference between any lottery $\mu$ and the status quo, i.e., $\mathbbm{E}_{\mu}[\alpha_v(x)]$, is single crossing in $v$.} \cref{ass2:single-crossing} has two consequences relevant to our analysis. \eqref{SC1:incompatible} states that if a left voter prefers to approve, then {\it all} right voters prefer to reject, and vice versa. \eqref{SC2:more-extreme}, on the other hand, states that if a left (right) voter prefers to approve, then {\it all} left (right) voters with bliss points further away from the status quo also prefer to approve. For convenience, I refer to such voters as more extreme:
\begin{definition}
    A left voter $w$ is \emph{more extreme} than a left voter $v$ if $w < v < 0$. A right voter $w$ is \emph{more extreme} than a right voter $v$ if $0 < v < w$.
\end{definition}

I illustrate my results for quadratic voter preferences, $\alpha_v(x) = -(v-x)^2 + (v-0)^2 = -x^2 + 2vx$, which is one example of standard preferences that satisfy all the assumptions. \cref{fig:Ls-preferences} illustrates the preferences of a quadratic voter $v < 0$.

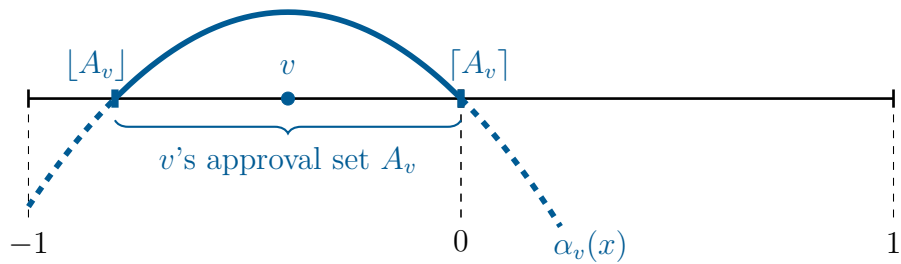
\begin{figure}[ht!]
    \centering
    \input{figures/fig-Ls-preferences}
    \caption{The policy outcome space $X = [-1,1]$, the status quo policy outcome $0$, a voter's bliss point $v < 0$, her net payoff from approval $\alpha_v(x) = -x^2 + 2vx$, and her approval set $A_v$. Under complete information, this voter prefers to approve policy outcomes left, but not too far left, of the status quo.}
    \label{fig:Ls-preferences}
\end{figure}

I focus on weak perfect Bayesian equilibria in which the voters' inference is consistent with disclosure on and off the path. In {\it equilibrium}, (i) the challenger sends messages that maximize his expected payoff; (ii) each voter approves whenever her expected net payoff from approval is non-negative under her posterior belief; (iii) each voter calculates her posterior using Bayes' rule on the equilibrium path; and (iv) a voter's posterior belief after an off-path message $m$ is an element of $\Delta m$. I restrict attention to equilibria in which all voters with bliss point $v \in X$ act the same. For ease of exposition, I also assume that status quo voters always reject.\footnote{A status quo voter's net payoff from approval is strictly negative unless $x=0$, so the only belief under which this voter weakly prefers to approve is $\delta_0$. Although there exist equilibria in which status quo voters approve if and only if $x = 0$, the prior measure of that event is zero since $\mu_0$ is atomless.} I denote the challenger's {\it interim} expected utility---i.e., the probability that he wins the election when his policy outcome is $x \in X$---by $U_I(x)$. I refer to the challenger's ex-ante utility as his {\it odds of winning}.

I say that a set of policy outcomes $W \subseteq X$ is {\it implementable} if there exists an equilibrium in which the challenger's interim utility is $U_I(x) = \mathbbm{1}(x \in W)$. In that equilibrium, the challenger wins if and only if $x \in W$ (i.e., $W$ is the set of ``winning'' policy outcomes) and his odds of winning are $\mu_0(W)$.

Most of my constructive results involve equilibria in which the challenger uses a simple pure strategy of revealing whether his policy outcome is or is not in some set.
Let $M:=(M_v)_{v \in V}$ be a collection of subsets of $X$.
A \emph{direct strategy} $\sigma_M$ is defined as $\sigma_{ M } (x) := ( M_v \text{ if } x \in M_v, \text{ otherwise } M_v^c )_{v \in V}$.
In words, when the challenger's policy outcome is $x$, he tells voter $v$ whether $x \in M_v$ (by sending message $M_v$) or not (by sending message $M_v^c$).
A {\it direct public strategy} $\sigma_M$ is one with $M_v = M$ for all $v \in V$.
When the challenger uses a direct strategy, voter $v$ hears one of two on-path messages, $m \in \{ M_v, M_v^c \}$, and thus learns whether $x \in m$ and nothing else; her posterior belief is then $\mu_0(\cdot \sep m)$.
I say that voter $v$ {\it is willing to approve} a set of policy outcomes $M_v$ (if the challenger uses a direct strategy $\sigma_M$ and $v$ receives message $M_v$) if it satisfies her obedience constraint:
\begin{equation}\label{obedience}\tag{obedience}
    \int_{M_v} \alpha_v(x) d\mu_0(x) \geq 0.
\end{equation}

\section{Analysis}\label{section:analysis}

To begin analysis, let us first classify elections in the following way. I say that coalition $D \subseteq V$ of voters is {\it left} ({\it right}), denoted $D < 0$ ($D > 0$), if it consists of left (right) voters only; I say that $D$ is a {\it mixed} coalition if it has both left and right voters.

\begin{definition}\label{dfn:election-classification}
    An election $(V,\mathcal{D})$ is a
    \begin{itemize}
        \item \ul{left- (right-) leaning election} if there exists a left (right) decisive coalition;
        \item \ul{status quo-leaning election} if every decisive coalition includes a status quo voter;
        \item \ul{polarized election} if there exists a mixed decisive coalition and there are no left or right decisive coalitions.
    \end{itemize}
    This classification is exhaustive: every election belongs to exactly one category.\footnote{By monotonicity and properness of $\mathcal{D}$, left and right decisive coalitions cannot coexist. If neither exists, then either (i) all decisive coalitions include a status quo voter (status quo-leaning election) or (ii) a mixed decisive coalition exists (polarized election). }
\end{definition}

Given a coalition of voters $D \subseteq V$, let $\lambda(D) := \max\limits_{v \in D,\, v < 0} v$ and $\rho(D) := \min\limits_{v \in D,\, v > 0} v$ be its \ul{left pivot} and \ul{right pivot}, if they exist, respectively. Let us define a representative voter as the voter whose approval set equals the union of the approval sets of all decisive coalitions.
\begin{definition}
    The \ul{representative voter} of election $(V,\mathcal{D})$ is the voter $v^* \in X$ such that $A_{v^*} = \bigcup\limits_{D \in \mathcal{D}} \bigcap\limits_{v \in D} A_{v}$.
\end{definition}

Voter $v^*$ essentially represents the preferences of the entire electorate because, under complete information, if she prefers to approve, then so does at least one decisive coalition, and if she prefers to reject, then so do all decisive coalitions. The lemma below determines the location of the representative voter in each election type, allowing us to state subsequent results for the public advertising game in terms of a single voter rather than sets of decisive coalitions.

\begin{lemma}\label{lemma:repr-voter}
    Consider an election $(V,\mathcal{D})$. Then,
    \[
        v^* =
        \begin{cases}
            \min\limits_{D \in \mathcal{D},\, D<0} \lambda(D) < 0, & \text{if it is a left-leaning election,}                    \\[1.2em]
            \max\limits_{D \in \mathcal{D},\, D>0} \rho(D) > 0,    & \text{if it is a right-leaning election,}                   \\[1.2em]
            0,                                                     & \text{if it is a status quo-leaning or polarized election.}
        \end{cases}
        ~\\[-\baselineskip]
    \]
\end{lemma}
\begin{proof}
    Recall that from \cref{ass1:single-peaked}, the approval set of voter $v \in V$ is $[ \floor{A_v},0 ]$ if $v < 0$, $A_v = \{ 0 \}$ if $v = 0$, and $A_v = [0, \ceil{A_v}]$ if $v > 0$. Therefore, for any decisive coalition $D \in \mathcal{D}$ that is mixed or includes a status quo, we have $\bigcap\limits_{v \in D} A_v = \{ 0 \}$. Consequently, for status quo-leaning and polarized elections, we have $v^* = 0$.

    Next, observe that from \eqref{SC2:more-extreme}, if voter $v$ is less extreme than voter $w$, then $A_v \subseteq A_w$. Then, for any left coalition $D<0$, we have $\bigcap\limits_{v \in D} A_v = \bigcap\limits_{v \in D} [ \floor{A_v},0] = [ \floor{A_{\lambda(D)}},0 ] = A_{ \lambda(D) }$, because $\lambda(D)$ is the least extreme voter in $D$. Therefore, for any left-leaning election $(V,\mathcal{D})$, we have
    \begin{align*}
        \bigcup\limits_{D \in \mathcal{D}} \bigcap\limits_{v \in D} A_v =
        \bigcup\limits_{\substack{D \in \mathcal{D} \\ D <0}} \bigcap\limits_{v \in D} A_v =
        \bigcup\limits_{\substack{D \in \mathcal{D} \\ D < 0}} [ \floor{A_{\lambda(D)}},0 ]
        = [ \floor{A_{v^*}},0 ]
        = A_{v^*},
    \end{align*}
    where $v^* = \min\limits_{D \in \mathcal{D},\, D<0} \lambda(D) < 0$. The proof is analogous for a right-leaning election.\qedhere
\end{proof}

\subsubsection*{Public Advertising}

In the public advertising game, the voters' common prior belief is updated to a common posterior. Therefore, the electorate faces a collective choice problem: whether to choose a safe option (the status quo) or a lottery over the challenger's policy outcomes (represented by their common posterior belief $\mu \in \Delta X$). The first result describes which elections are ``unwinnable'' for the challenger with public advertising.

\begin{theorem}\label{thm1:PA} In the public advertising game,
    \begin{enumerate}
        \item The challenger's odds of winning are zero in every equilibrium if and only if the election is status quo-leaning or polarized, i.e., there is no left or right decisive coalition.
        \item A set of policy outcomes is implementable in a left/right-leaning election if it is obedient for the representative voter and includes her approval set.
    \end{enumerate}
\end{theorem}
As the equilibrium constructions from the proofs of Theorems \ref{thm1:PA} and \ref{thm2:TA} are used later in \cref{section:baseline}, I provide the proofs of these theorems in the main text.
\begin{proof}

    I prove sufficiency of part 1 directly. If there is no left or right decisive coalition, then the election is either status quo-leaning or polarized. In both cases, the challenger convinces a decisive coalition only if the public belief is $\delta_0$, or whenever $x = 0$, which has zero prior measure. To see this, observe that for any belief other than $\delta_0$, status quo voters strictly prefer to reject; also, left and right voters never prefer to approve at the same time by \ref{SC1:incompatible}. Therefore, the challenger's odds of winning are zero in every equilibrium of the public advertising game unless there is a left or right decisive coalition.

    To prove the rest of the theorem, we consider (without loss) a left-leaning election with representative voter $L<0$ and construct an equilibrium that implements $M \subseteq X$, where $A_L \subseteq M$ and $\int_M \alpha_L(x) d\mu_0(x) \geq 0$. Note that such $M$ exists---for example, we can let $M = A_L = [ \floor{A_L},0 ]$.

    Suppose that the challenger uses the direct public strategy $\sigma_M$. When message $M$ is heard on the path, the public posterior becomes $\mu_0 (\cdot \sep M)$; the representative voter prefers to approve because $M$ satisfies her obedience constraint. Since $L$ is the least extreme voter in some left decisive coalition $D \in \mathcal{D}$, by \eqref{SC2:more-extreme} we have $\int_M \alpha_v(x) d\mu_0(x) \geq 0$ for all $v \in D$, which means that the {\it entire} decisive coalition prefers to approve after message $M$. 

    We now show that when the public posterior $\mu$ is supported on $M^c$ (for instance, after message $M^c$), the challenger loses the election because {\it every} decisive coalition includes a voter who strictly prefers to reject. First, voter $L$ strictly prefers to reject. Indeed, since $A_L \subseteq M$, we have $\alpha_L(x) < 0$ for all $x \in M^c$, and therefore $\mathbbm{E}_\mu[\alpha_L (x)] < 0$. Next, all left voters less extreme that $L$ also strictly prefer to reject: by the contrapositive of \eqref{SC2:more-extreme}, $\mathbbm{E}_\mu[\alpha_L (x)] < 0$ implies $\mathbbm{E}_\mu[\alpha_v (x)] < 0$ for all $v \in (L,0)$. In the left-leaning election under consideration, every decisive coalition is left, mixed, or includes a status quo voter. Every left decisive coalition includes a voter $v \in [L,0)$ by \cref{lemma:repr-voter}. Every other decisive coalition contains a voter who strictly prefers to reject by the argument in the first paragraph of this proof. Consequently, the challenger loses whenever the public belief is supported on $M^c$ (in particular, when message $M^c$ is heard on the path). Therefore, if the challenger uses strategy $\sigma_M$, he wins if and only if $x \in M$, his interim utility is $U_I(x) = \mathbbm{1}(x \in M)$, and his odds of winning are $\mu_0(M) > 0$.

    Finally, we specify the voters' off-path beliefs and show that the challenger does not have profitable deviations. When $x \in M$, the challenger receives the highest possible payoff and, therefore, has no profitable deviations. We thus do not restrict voters' beliefs for off-path messages $m \subset M$. Conversely, when the challenger's policy outcome is not in $M$, his on-path payoff is zero; a profitable deviation would require sending a message that results in a win. To deter such deviations, we require voters to hold skeptical beliefs. Specifically, for any off-path message $m \not\subset M$, we restrict the voters' posterior to be an element of $\Delta(m \cap M^c)$. This ensures that any message that the challenger can send when $x \notin M$ induces a public posterior supported on $M^c$, after which the challenger loses the election, as shown in the previous paragraph. This completes the equilibrium characterization and the proof.  
\end{proof}

\cref{thm1:PA} essentially states that an election is winnable with public advertising if and only if the representative voter is left or right. Furthermore, in winnable elections, the challenger effectively caters to the representative voter: a set of ``winning'' policy outcomes is implementable if it contains that voter's approval set and satisfies her obedience constraint. The former condition ensures that the representative voter (and thus some voter in every decisive coalition) strictly prefers to reject after {\it every} message available to the challenger when $x \notin M$. The latter condition guarantees that the representative voter (and thus every voter in some decisive coalition) approves after message $M$. In \cref{section:baseline}, we will characterize the largest implementable set of winning policy outcomes for a left-leaning election in which all voters' bliss points are the same.

Under simple majority, we arrive at a familiar characterization of elections that are unwinnable with public advertising.

\begin{corollary}\label{cor1:MVT}
    Under simple majority, the challenger's odds of winning are zero in every equilibrium of the public advertising game if and only if the status quo is the median voter's bliss point.
\end{corollary}

The proof of the corollary is straightforward: if there is no left or right decisive coalition, then the status quo is the median voter's bliss point. Note that \cref{cor1:MVT} is a special case of median voter theorems for collective choice problems under uncertainty, which state that the median voter's bliss point defeats any lottery, degenerate or nondegenerate. The median voter theorem holds for single-peaked and strictly concave voter preferences (\citenop{Shepsle1972}), and those that have single-crossing expectational differences (\citenop{Kartik2023}).\footnote{Note that if $x=0$ is the median voter's bliss point, then it defeats any lottery (degenerate or non-degenerate) if $\alpha_v(x)$ satisfies \eqref{SC1:incompatible}---see the first paragraph of the proof of \cref{thm1:PA}. If $\alpha_v$ is single-peaked and strictly concave, then \eqref{SC1:incompatible} follows from Jensen's inequality; \eqref{SC1:incompatible} is also implied by single-crossing expectational differences.}

\subsubsection*{Targeted Advertising}

Some elections are unwinnable for the challenger with public advertising because the status quo beats any lottery over the challenger's policy outcomes. Targeted advertising allows the challenger to induce different beliefs among different voters and win some of these elections. The next result describes which elections are ``unwinnable'' for the challenger with targeted advertising.

\begin{theorem}\label{thm2:TA} In the targeted advertising game,
    \begin{enumerate}
        \item The challenger's odds of winning are zero in every equilibrium if and only if the election is status quo-leaning, i.e., every decisive coalition includes a status quo voter.
        \item An interval $[a,b]$ such that $a<0<b$ is implementable in a polarized election with a mixed decisive coalition $D \in \mathcal{D}$ if the set $[a,b] \cup A_v$ is obedient for each pivot $v \in \{ \lambda(D), \rho(D) \}$.
    \end{enumerate}
\end{theorem}
\begin{proof}
    For sufficiency of part 1, recall that status quo voters always reject.
    I prove necessity of part 1 by contraposition.
    If not every decisive coalition includes a status quo voter, then two cases are possible: the election is left- (right-) leaning, or it is polarized.
    In the first case, the public advertising equilibria described in the proof of \cref{thm1:PA} are equilibria of the targeted advertising game.
    In the remainder of this proof, we focus on a polarized election with a mixed decisive coalition $D \in \mathcal{D}$; for brevity, we denote its pivots by $L := \lambda(D)$ and $R := \rho(D)$. We will show that an interval described in part 2 of \cref{thm2:TA} exists and construct an equilibrium that implements it.

    \paragraph{Preliminaries.}  Let $[a,b]$ be an interval such that $a < 0 < b$ and the set $A_v \cup [a,b]$ satisfies the obedience constraint of each pivot $v \in \{ L,R \}$. To see why such an interval exists, observe that voter $L$ ($R$) is willing to approve some right (left) policy outcomes, as long as her expected net payoff from approval is non-negative given her belief. Mathematically, there exist $\widetilde{a} < 0$ and $\widetilde{b} > 0$ such that $\int_{\floor{A_L}}^{\widetilde{b}} \alpha_L(x) d\mu_0(x) \geq 0$ and $\int_{\widetilde{a}}^{\ceil{A_R}} \alpha_R(x) d\mu_0(x) \geq 0$.\footnote{Such $\widetilde{b} > 0$ exists because the function $\phi_L(z):= \int_{ \floor{A_L} }^{z} \alpha_L(x)d\mu_0(x)$ is continuous, strictly decreasing for $z \geq 0$, and $\phi_L(0) > 0$. By a similar argument, there exists $\widetilde{a} < 0$ such that $\int_{\widetilde{a}}^{ \ceil{A_R} } \alpha_R(x)d\mu_0(x) \geq 0$.}

    \paragraph{On-path behavior.} Let the challenger use a direct strategy $\sigma_{(M_v)}$, where $M_v = A_v \cup [a,b]$ for all $v \in V$. Then, pivot $v \in \{ L,R \}$ approves after $M_v$, because it satisfies her obedience constraint by construction. Every left (right) voter $v \in [-1,L)$ ($v \in (R,1]$) more extreme than $L$ ($R$) also approves after $M_v$.\footnote{I formally prove this statement in \cref{lemma:SC2-appendix} in the appendix; it follows from \eqref{SC2:more-extreme}.} Voters with bliss points in $(L,R)$ approve or reject after $M_v$ depending on the sign of $\int_{M_v} \alpha_v(x) d\mu_0(x)$. Every voter $v \in V$ rejects after $M_v^c$ because $\alpha_v(x) < 0$ for all $x \in M_v^c$.

    \paragraph{Skeptical off-path beliefs.} For any off-path message $m \notin \{ M_v, M_v^c \}$, let voter $v$'s posterior be an element of $\Delta (m \cap A_v^c)$ whenever $m \cap A_v^c$ is non-empty. That way, when a voter $v\in V$ hears an off-path message that includes {\it any} policy outcome outside of her approval set, she rejects.

    \paragraph{Electoral outcome.} The challenger convinces {\it every} voter in the decisive coalition $D$ if $x \in \bigcap\limits_{v \in D} M_v = [a,b]$. In the polarized election under consideration, every decisive coalition is mixed or includes a status quo voter. Since left (right) voter never approve after a message including policy outcomes right of $b$ (left of $a$), and status quo voters always reject, the challenger loses if $x \notin [a,b]$. His interim utility is thus $U_I(x) = \mathbbm{1}(x \in [a,b])$, and his odds of winning are $\mu_0([a,b])>0$.

    \paragraph{No profitable deviations.} If $x \in [a,b]$, the challenger receives the highest possible payoff and, therefore, has no profitable deviations. If $x < a < 0$, then {\it any} verifiable message (i.e., a message that includes $x$) convinces {\it every} right voter to reject. Since every decisive coalition includes a right voter or a status quo voter, the challenger has no profitable deviations. A similar argument applies when $x > b > 0$. This completes the equilibrium characterization and the proof.
\end{proof}

\cref{thm2:TA} describes implementable sets of winning policy outcomes in the targeted advertising game for polarized elections. In the proof, I show how to implement these outcomes in an equilibrium where the challenger uses a direct strategy that effectively caters to the left and right pivots, $L$ and $R$, of some mixed decisive coalition. More precisely, an interval $[a,b]$ is implementable if the set $A_v \cup [a,b]$ is obedient for pivots $L$ and $R$---that is, the lower bound $a<0$ cannot be too small, and the upper bound $b>0$ cannot be too large. In \cref{section:baseline}, we will characterize the largest implementable interval of winning policy outcomes in a polarized election in which all left (right) voters share the same bliss point.

Theorems \ref{thm1:PA} and \ref{thm2:TA} describe how the challenger advertises his policy outcome depending on the composition of the electorate. If every decisive coalition includes a status quo voter, he loses with public and targeted advertising. If the election is left- (right-) leaning, he wins by advertising publicly and tailoring his strategy to the representative voter. Finally, if the election is polarized, then the challenger can win with targeted advertising but not public advertising. Under simple majority, targeted advertising allows the challenger to defeat the status quo policy that the median voter theorems deem unbeatable.

\section{Baseline Election}\label{section:baseline}

While Theorems \ref{thm1:PA} and \ref{thm2:TA} characterize which elections are winnable with public disclosure and targeted advertising, they do not make a unique prediction about how the challenger wins these elections. Specifically, the proof of each theorem involves providing an example of an equilibrium in which the challenger's odds of winning are positive. The reason for this is that the model admits multiple equilibria. There are two sources of multiplicity. First, there may be multiple decisive coalitions. Second, the verifiable disclosure game has a range of equilibrium outcomes even if there is only one receiver (\citenop{TitovaZhangPVI}). To move forward in the analysis, I consider a class of baseline elections in which the minimal decisive coalition is unique. I focus on the challenger-preferred equilibrium in order to provide the upper bound on the challenger's odds of winning across all equilibria.

\begin{definition}
    A \ul{baseline election} has electorate $\{ L,0,R \}$, where $-1 \leq L < 0 < R \leq 1$.
\end{definition}

In the baseline election, all left voters have the same bliss point, $L < 0$, and all right voters have the same bliss point, $R > 0$. This assumption limits the number of possible decisive coalitions and allows us to focus on the messages to be sent to left (right) voters, all of whom have the same bliss point.

The baseline election serves as a building block for a general (non-baseline) election $(V,\mathcal{D})$.
Specifically, \cref{thm1:PA} and its proof show how to implement any set of winning policy outcomes that is obedient for the representative voter $v^*$ in the public advertising game for a left- (right-) leaning election.
Consequently, the set of winning policy outcomes for a {\it left- (right-) leaning baseline election} (described in the subsequent \cref{prop:baseline-PA}) remains implementable in the more general election with representative voter $v^* \in \{ L,R \}$. When it comes to polarized elections, which become winnable with targeted advertising, \cref{thm2:TA} and its proof describe how to implement an interval of winning policy outcomes that is obedient for left and right pivots of some mixed decisive coalition $D \in \mathcal{D}$.
Consequently, the interval of winning policy outcomes that we find for a {\it polarized baseline election} (see \cref{prop:baseline-TA}) remains implementable in a more general election that has a mixed decisive coalition with pivots $L$ and $R$.\footnote{Note that a polarized general election may have multiple mixed decisive coalitions with distinct pairs of pivots. Such an election consequently would have multiple corresponding baseline elections (one for each distinct pair of pivots) and thus multiple implementable intervals of winning policy outcomes. While finding the largest implementable interval of winning policy outcomes for a polarized general election is beyond the scope of this paper, \cref{prop:comp_stats} shows that tailoring one's messages to more polarized decisive coalitions increases the challenger's odds of winning.}

Given our focus on the challenger-preferred equilibrium, it is useful to first find the highest probability of convincing one voter. Consider the following auxiliary problem with parameters $l$ and $r$ such that $-1 \leq l \leq \floor{A_v} < \ceil{A_v} \leq r \leq 1$:
\begin{equation}\label{eqn:aux_prob}\tag{AUX}
    \max_{I \subseteq [l,r]} \mu_0(I) \quad \text{subject to} \quad \int_{I} \alpha_v(x) d\mu_0(x) \geq 0.
\end{equation}
Roughly speaking, Problem \eqref{eqn:aux_prob} identifies the largest (in terms of prior measure) set of policy outcomes $I$ such that, if voter $v$ learns that $x \in I$ and nothing else (i.e., if the challenger uses a direct strategy $\sigma_I$), she prefers to approve. The objective function is the probability of convincing the voter, while the obedience constraint $\int_{I} \alpha_v(x) d\mu_0(x) \geq 0$ ensures that the voter prefers to approve given her information.
Parametrizing the problem with $l$ and $r$ allows us to focus on certain subsets of $X$.
For example, if we are interested in the largest set of right policy outcomes that a left voter $L < 0$ is willing to approve, then we let $l = \floor{A_L}$ and $r = 1$.

Problem \eqref{eqn:aux_prob} comes from the information design literature and provides a theoretical upper bound on the ex-ante probability of convincing a Bayesian voter (see, e.g., \citenop{Alonso2016} and \citenop{TitovaZhangPVI}).
The solution is an interval characterized by a cutoff value for the voter's net payoff from approval: voter $v$ approves every policy outcome $x \in [l,r]$ for which $\alpha_v(x) \geq c_v^*(l,r)$, that is, the net payoff is not too negative. The cutoff value $c_v^*(l,r) < 0$ is obtained from the binding obedience constraint. The set $\{ x \in [l,r] \sep \alpha_v(x) \geq c_v^*(l,r) \}$ is the upper contour set of the strictly quasiconcave function $\alpha_v(x)$ (by \cref{ass1:single-peaked}) and is therefore an interval. \cref{lemma:aux-solution} characterizes the solution of the auxiliary problem; the formal proof can be found in the appendix.

\begin{lemma}\label{lemma:aux-solution}
    For \( v \neq 0 \) and parameters \( -1 \leq l \leq \lfloor A_v \rfloor < \lceil A_v \rceil \leq r \leq 1 \), Problem \eqref{eqn:aux_prob} admits a solution \( I_v(l,r) \) that is a closed interval. Specifically,
    \begin{itemize}
        \item if \( \displaystyle\int_l^r \alpha_v(x) \, d\mu_0(x) \geq 0 \), then \( I_v(l,r) = [l, r] \),
        \item if \( \displaystyle\int_l^r \alpha_v(x) \, d\mu_0(x) < 0 \), then \( I_v(l,r) = \{ x \in [l, r] \sep \alpha_v(x) \geq c_v^*(l,r) \} \), where \( c_v^*(l,r) < 0 \) is the unique value satisfying \( \displaystyle\int_{I_v(l,r)} \alpha_v(x) \, d\mu_0(x) = 0\).
    \end{itemize}
    Furthermore, all other solutions coincide with \( I_v(l,r) \) \(\mu_0\)-almost everywhere.
\end{lemma}

\subsection*{Winning Elections with Public Advertising}

Here I find the challenger-preferred equilibrium for a left/right-leaning baseline elections, which are winnable with public advertising. Without loss of generality, I will focus on a left-leaning election, meaning that $\{ L \} \in \mathcal{D}$. In this election, targeted advertising is as good as public advertising, because the challenger wins if and only if voter $L$ approves. To maximize his odds of winning, the challenger finds the largest subset of $[-1,1]$ that voter $L$ is willing to approve: that is, he solves Problem \eqref{eqn:aux_prob} for voter $L$ with parameters $l=-1$ and $r=1$. He then publicly reveals whether his policy outcome is in that interval.

\begin{proposition}\label{prop:baseline-PA}
    Consider a left-leaning baseline election, $\{L\} \in \mathcal{D}$. Then the challenger's highest odds of winning across all equilibria of the public advertising and targeted advertising games are $\mu_0 ( I_L(-1,1) )$. He achieves these odds by publicly revealing to all voters whether his policy outcome is in $I_L(-1,1)$.
\end{proposition}
\begin{proof}
    In a left-leaning baseline election, voter $L$ is trivially the representative voter. By \cref{lemma:aux-solution}, $I_L(-1,1)$ includes her approval set and satisfies her obedience constraint. Using the equilibrium construction from the proof of \cref{thm1:PA}, we conclude that the set $I_L(-1,1)$ is implementable; in the equilibrium that implements it, the challenger uses a direct public strategy $\sigma_M$ with $M = I_L(-1,1)$, his interim utility is $U_I(x) = \mathbbm{1}(x \in I_L(-1,1))$, and his odds of winning are $\mu_0(I_L(-1,1)) > 0$. His odds of winning cannot be higher in any other equilibrium because $\mu_0(I_L(-1,1))$ is the upper bound on the odds of convincing a Bayesian voter $L$.
\end{proof}

I illustrate the equilibrium outcome in \cref{fig:baseline-public}.\footnote{\cref{fig:baseline-public} presents the numerical solution $I_L(-1,1) = [-0.82,0.22]$ for a quadratic voter $L = -0.3$ and a uniform prior.}
Note that the challenger's equilibrium strategy described in \cref{prop:baseline-PA} remains an equilibrium strategy of the public advertising game for a left-leaning non-baseline election with representative voter $L$. See the proof of \cref{thm1:PA} for a full description of this equilibrium.
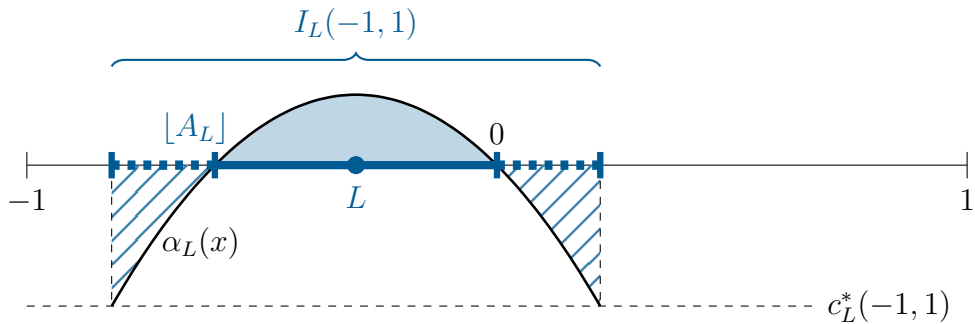
\begin{figure}[h!]
    \centering
    \input{figures/fig-baseline-public}
    \caption{To maximize his odds of convincing the decisive coalition $\{L\}$, the challenger reveals whether his policy outcome is in $I_L(-1,1)$. Under a uniform prior, $c_L^*$ is obtained from equating the solid area to the dashed area so that voter $L$ is indifferent between approval and rejection when she learns that $x \in I_L(-1,1)$.}
    \label{fig:baseline-public}
\end{figure}

\subsection*{Swinging Elections with Targeted Advertising}

In the remainder of this section, I focus on a baseline election that is unwinnable with public advertising but winnable with targeted advertising. From Theorems \ref{thm1:PA} and \ref{thm2:TA}, that election is polarized, so the unique decisive coalition that does not include a status quo voter is $\{ L,R \}$. The challenger wins if and only if left and right voters approve. Let us consider the following problem:
\begin{equation}\label{problem:AUX-TA}\tag{AUX-TA}
    \begin{aligned}
         & \max_{(M_L, M_R) \subseteq X^2}\mu_0(M_L\cap M_R)
        \quad \text{subject to}
        \\
         & \quad \int_{M_v} \alpha_v(x) \, d\mu_0(x) \geq 0 \quad \text{for each } v \in \{L, R\}.
    \end{aligned}
\end{equation}
I will shortly show that this problem admits a solution $(M_L^*, M_R^*)$ such that $W^* \coloneqq M_L^* \cap M_R^*$ is an interval $[a, b]$ with $a < 0 < b$. Using the equilibrium construction from \cref{thm2:TA}, I will conclude that the set $W^*$ is implementable; in the equilibrium that implements it, the challenger uses a direct strategy $\sigma_{(M_L^*,M_R^*)}$, his interim utility is $U_I(x) = \mathbbm{1}(x \in W^*)$ and his odds of winning are $\mu_0(W^*)$.

Then, I will show that the pair $(M_L^*,M_R^*)$ also characterizes an {\it optimal experiment} that solves the information design problem, in which the challenger chooses and commits to an experiment prior to learning his policy outcome.\footnote{I formalize the challenger's information design problem in the appendix.}
Consequently, $\mu_0(W^*)$ is the upper bound on the odds of convincing Bayesian voters $L$ and $R$, and thus the highest odds of winning across {\it all} equilibria of the targeted advertising game.
To understand why the challenger reaches this upper bound in an equilibrium of a verifiable information game, observe that Problem \eqref{problem:AUX-TA} imposes only the voters' obedience constraints.
Relative to an optimal experiment, an equilibrium strategy must satisfy additional incentive-compatibility constraints for the challenger (i.e., he must not have profitable deviations from an equilibrium strategy for each realized policy outcome $x \in X$).
However, given that his objective is to convince $L$ and $R$ to approve and that his messages must be verifiable, $\sigma_{(M_L^*,M_R^*)}$ automatically satisfies these constraints; see the equilibrium construction in the proof of \cref{thm2:TA}.

Let us now show that Problem \eqref{problem:AUX-TA} admits a solution $(M_L^*, M_R^*)$ such that $M_L^* \cap M_R^*$ is a closed interval. For that, it is helpful to define the largest asymmetric interval of policy outcomes that each voter is willing to approve:
\begin{definition}\label{dfn:largest-interval-AS}
    The \emph{largest asymmetric interval of approved policy outcomes} $I_v$ for voter $v \in \{L, R\}$ is defined as follows:

    \medskip

    \begin{tabular}{@{}l}
        For $v = L < 0$: \quad $I_L = [\, \lfloor A_L \rfloor,\, b_L\, ] := I_L(\lfloor A_L \rfloor,\, 1),$                                       \\
        \quad\quad where $I_L(\lfloor A_L \rfloor,\, 1)$ solves Problem~\eqref{eqn:aux_prob} for $L$ with $l = \lfloor A_L \rfloor$ and $r = 1$. \\
        \\
        For $v = R > 0$: \quad $I_R = [\, a_R,\, \lceil A_R \rceil\, ] := I_R(-1,\, \lceil A_R \rceil),$                                          \\
        \quad\quad where $I_R(-1,\, \lceil A_R \rceil)$ solves Problem~\eqref{eqn:aux_prob} for $R$ with $l = -1$ and $r = \lceil A_R \rceil$.
    \end{tabular}

\end{definition}

Simply put, $I_L$ includes $L$'s approval set $[\floor{A_L}, 0]$ and as many right policy outcomes $(0, b_L]$ as her obedience constraint allows. Similarly, $I_R$ includes $R$'s approval set $[0, \ceil{A_R}]$ and the largest set $[a_R, 0)$ of left policy outcomes satisfying her obedience constraint. \cref{fig:AS} illustrates these intervals for quadratic voters $L = -0.2$ and $R = 0.25$, and a uniform prior.\footnote{Here, $I_L = [-0.4, 0.2]$ and $I_R = [-0.25, 0.5]$. Notably, when preferences are quadratic and $\mu_0$ is uniform, for $L \geq -0.5$ we have $I_L = [2L,-L]$, and for $R \leq 0.5$ we have $I_R = [-R,2R]$.}

\begin{figure}[H]
    \centering
    \begin{subfigure}[b]{0.495\textwidth}
        \centering
        \input{figures/fig-largest-interval-AS-L}
        \caption{$I_L = [\floor{A_L},b_L ]$.}
        \label{Fig:SubLeft}
    \end{subfigure}%
    ~
    \begin{subfigure}[b]{0.495\textwidth}
        \centering
        \input{figures/fig-largest-interval-AS-R}
        \caption{$I_R = [a_R, \ceil{A_R} ]$.}
        \label{Fig:SubRight}
    \end{subfigure}
    \caption{Largest asymmetric intervals of approved policy outcomes of voters $L$ and $R$. Under uniform prior, $b_L$ and $a_R$ are obtained from equating the solid and dashed areas.}
    \label{fig:AS}
\end{figure}
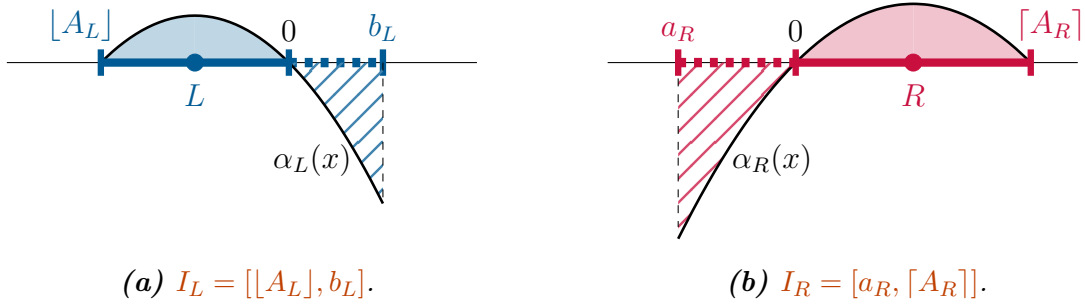

One might guess that sending each voter her largest asymmetric interval---specifically, setting $(M_L^*,M_R^*) = (I_L,I_R)$---maximizes the challenger's odds of winning. This strategy is indeed optimal if $I_L \cap I_R = [a_R, b_L]$, which occurs when $\floor{A_L} \leq a_R$ and $b_L \leq \ceil{A_R}$. It is straightforward to see that the challenger's odds of winning cannot exceed $\mu_0([a_R, b_L])$: every policy outcome outside $[a_R, b_L]$ lies further from at least one voter's bliss point, making it more ``costly'' in terms of that voter's obedience constraint. \cref{fig:baseline-TA-interior} illustrates this challenger-preferred equilibrium outcome for quadratic voters $L = -0.2$ and $R = 0.25$, and a uniform prior.

\begin{figure}[h!]
    \centering
    \input{figures/fig-baseline-TA-interior}
    \caption{The electoral outcome when the challenger reveals to left voters whether his policy outcome is in $[\floor{A_L},b_L]$ and to right voters whether his policy outcome is in $[a_R,\ceil{A_R}]$. The decisive coalition $\{ L,R \}$ approves policy outcomes in $[a_R,b_L]$.}
    \label{fig:baseline-TA-interior}
\end{figure}
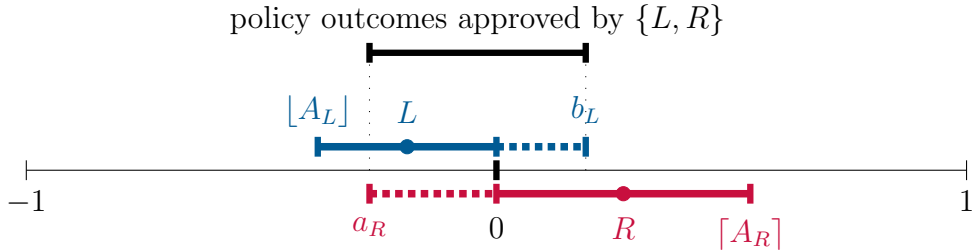

Next, consider the case when $a_R < \floor{A_L}$ and $b_L \leq \ceil{A_R}$. It is clear that sending each voter her largest asymmetric interval no longer maximizes the challenger's odds of winning. Indeed, $R$ is now willing to approve $L$'s entire approval set as well as the policy outcomes left of $\floor{A_L}$, which left voters prefer to policy outcomes close to $b_L$. Hence, $M_L^*$ should start at $a_R$ and span as far right as possible. Formally, in this case, $M_L^* = I_L( a_R, 1 )$ and $M_R^* = I_R$ so that the challenger wins whenever $x \in I_L( a_R, 1 )$.\footnote{Note that if $a_R \leq \floor{I_L(-1,1)}$, then $I_L(a_R,1) = I_L(-1,1)$, since $I_L(-1,1)$ is the largest set that voter $L$ is willing to approve. 
Otherwise, if $\floor{I_L(-1,1)} < a_R < \floor{A_L}$, then $\floor{I_L(a_R,1)} = a_R$. For details, see cases 2.1 and 2.2 (respectively) of the proof of \cref{prop:baseline-TA} in the appendix.} \cref{fig:baseline-TA-exterior} illustrates this outcome for quadratic voters $L = -0.2$, $R = 0.45$, and a uniform prior. The case when $\floor{A_L} \leq a_R$ and $\ceil{A_R} < b_L$ is symmetric.

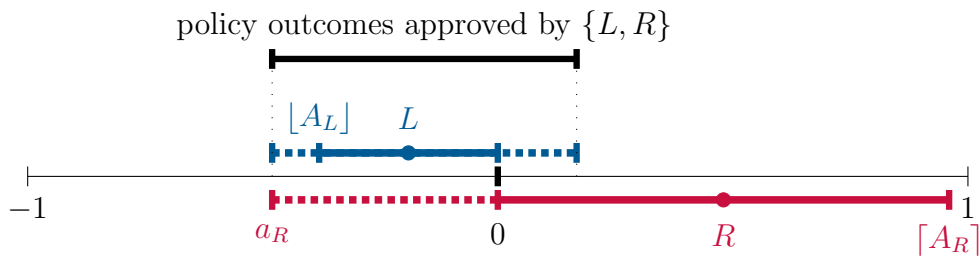
\begin{figure}[h!]
    \centering
    \input{figures/fig-baseline-TA-exterior}
    \caption{ To maximize the odds of convincing the decisive coalition $\{ L,R \}$ when $a_R < \floor{A_L}$, the challenger convinces left voters to approve the largest subset of $[a_R,1]$.}
    \label{fig:baseline-TA-exterior}
\end{figure}

The final case, $a_R < \floor{A_L}$ and $\ceil{A_R} < b_L$, is impossible: these conditions would imply that voters $L$ and $R$ both prefer to approve under a common belief $\mu_0(\cdot \sep [ \floor{A_L}, \ceil{A_R} ])$, contradicting \eqref{SC1:incompatible}.\footnote{By \cref{dfn:largest-interval-AS}, we have $\int_{a_R}^{\ceil{A_R}}\alpha_R(x) d\mu_0(x) \geq 0$. Then, $a_R < \floor{A_L}$ implies that $\int_{\floor{A_L}}^{\ceil{A_R}}\alpha_R(x) d\mu_0(x) \geq 0$, i.e., $R$ approves under belief $\mu_0(\cdot \sep [\floor{A_L},\ceil{A_R}])$. Similarly, if $\ceil{A_R} < b_L$, then $L$ approves under the same belief.} The following proposition summarizes the targeted advertising strategy that maximizes the challenger's odds of winning a polarized baseline election.

\begin{proposition}\label{prop:baseline-TA}
    Consider a polarized baseline election, $\{L,R\} \in \mathcal{D}$ but $\{L\},\{R\} \notin \mathcal{D}$. Then, the challenger’s maximal odds of winning across all equilibria of the targeted advertising game are $\mu_0(M_L^* \cap M_R^*) > 0$:

    \begin{enumerate}[itemsep=0pt]
        \item If $\floor{A_L} \leq a_R$ and $b_L \leq \ceil{A_R}$, then
              \[
                  M_L^*            = [\floor{A_L}, b_L], \quad
                  M_R^*            = [a_R, \ceil{A_R}], \quad
                  M_L^* \cap M_R^* = [a_R, b_L].
              \]~\\[-2\baselineskip]
        \item If $a_R < \floor{A_L}$ and $b_L \leq \ceil{A_R}$, then
              \[
                  M_L^*            = I_L(a_R, 1), \quad
                  M_R^*            = [a_R, \ceil{A_R}], \quad
                  M_L^* \cap M_R^* = M_L^*.
              \]~\\[-2\baselineskip]
        \item If $\floor{A_L} \leq a_R$ and $\ceil{A_R} < b_L$, then
              \[
                  M_L^*            = [\floor{A_L}, b_L], \quad
                  M_R^*            = I_R(-1, b_L), \quad
                  M_L^* \cap M_R^* = M_R^*.
              \]~\\[-2\baselineskip]
    \end{enumerate}
    The challenger achieves these odds of winning by using the direct strategy $\sigma_{(M_L^*,M_R^*)}$.
\end{proposition}

The formal proof of this result is in the appendix and involves three steps. In step 1, I confirm that the pair $(M_L^*, M_R^*)$ solves Problem \eqref{problem:AUX-TA}.
In step 2, I describe the equilibrium using the construction from the proof of \cref{thm2:TA}.
In step 3, I formulate the information design problem and show that the pair $(M_L^*,M_R^*)$ also characterizes an optimal experiment and the challenger's odds of winning are also $\mu_0(M_L^* \cap M_R^*)$.

\subsubsection*{Comparative Statics}

\cref{prop:baseline-TA} suggests that the boundaries of the challenger-preferred equilibrium set of winning policy outcomes depend on the voters' bliss points. I explore this relationship here. First, observe that more extreme voters are willing to approve wider ranges of policy outcomes.

\begin{lemma}\label{lemma:more-extreme-more-persuadable}
    If $w$ is a more extreme voter than $v$, then $I_w \supseteq I_v$. Furthermore,
    \begin{itemize}
        \item if $0<v<w$, then $\floor{I_w} = a_w \leq a_v = \floor{I_v}$, and the inequality is strict if $-1 < a_v$ and $ \ceil{A_v} < \ceil{A_w} $;
        \item if $w<v<0$, then $\ceil{I_v} = b_v \leq b_w = \ceil{I_w}$, and the inequality is strict if $b_v < 1$ and $\floor{A_w} < \floor{A_v}$.
    \end{itemize}
\end{lemma}

\cref{lemma:more-extreme-more-persuadable} follows from the assumption \eqref{SC2:more-extreme} that more extreme voters are more persuadable. \cref{fig:R-more-extreme} illustrates the intuition for \cref{lemma:more-extreme-more-persuadable} for right voters $v < w$. First, observe that a more extreme voter $w$ would be persuaded by the most biased message intended for voter $v$, that is, the interval $[a_v, \ceil{A_v}]$ is obedient for $w$. Furthermore, a more extreme voter $w$ has a larger approval set (i.e., $\ceil{A_v} \leq \ceil{A_w}$), so the interval $[a_v, \ceil{A_w}]$ is obedient for $w$ also.
Therefore, we can decrease the left boundary of that interval $a_v$ to $a_w \leq a_v$ so that the set $[a_w, \ceil{A_w}]$ remains obedient for $w$; the inequality is strict if $-1 < a_v$ and $ \ceil{A_v} < \ceil{A_w}$.

\begin{figure}[h!]
    \centering
    \input{figures/fig-R-more-extreme}
    \caption{A more extreme voter $w$ has a larger asymmetric interval of approved policy outcomes than a less extreme voter $v$.}
    \label{fig:R-more-extreme}
\end{figure}
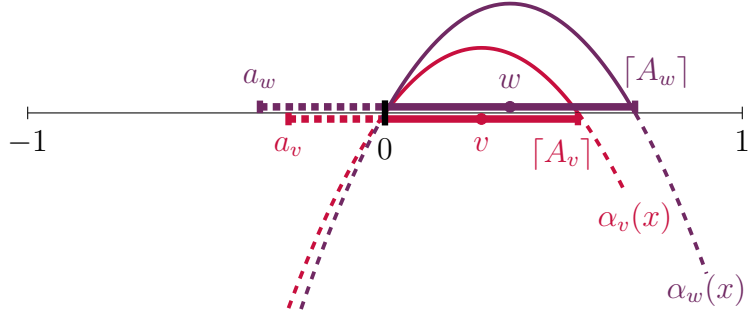

Next, let us explore how the challenger-preferred equilibrium outcome of a baseline election, described in \cref{prop:baseline-TA}, changes as right voters become more extreme. Note that making right voters more extreme makes the (already polarized) electorate more polarized (\citenop{EstebanRay1994}).

\begin{proposition}\label{prop:comp_stats}
    Consider the targeted advertising game with a polarized baseline election, $\{L,R\} \in \mathcal{D}$ but $\{L\},\{R\} \notin \mathcal{D}$. Let $(M_L^*,M_R^*)$ be the challenger-preferred equilibrium intervals of approved policy outcomes described in \cref{prop:baseline-TA}. Suppose that $b_L \leq \ceil{A_R}$. Then, as $R$ increases,
    \begin{itemize}
        \item the challenger's odds of winning $\mu_0(M_L^* \cap M_R^*)$ increase;
        \item the set of winning policy outcomes $M_L^* \cap M_R^*$ shifts to the left, that is, $\floor{ M_L^* \cap M_R^* }$ and $\ceil{ M_L^* \cap M_R^* }$ decrease.
    \end{itemize}
\end{proposition}

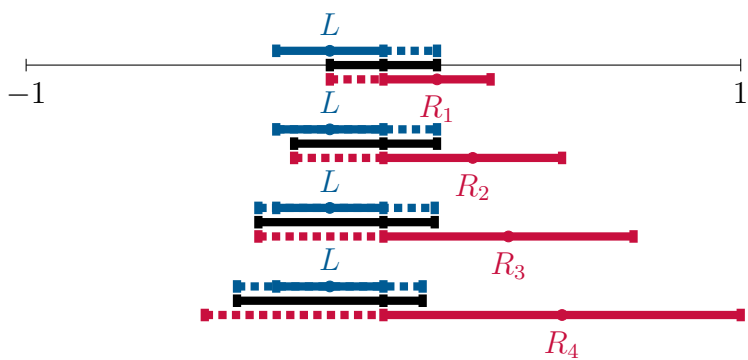
\begin{figure}[h!]
    \centering
    \input{figures/fig-comp_statics}
    \caption{The challenger-preferred equilibrium outcome as right voters become more extreme (top to bottom). Right voters approve ranges of policy outcomes (in red) that span further left, and the set of winning policy outcomes (in black) shifts left.}
    \label{fig:comp_statics}
\end{figure}

\cref{fig:comp_statics} illustrates the equilibrium outcomes of four baseline elections, holding the left voters' bliss point $L$ fixed and increasing the right voters' bliss point from $R_1$ to $R_4$ (top to bottom).\footnote{\cref{fig:comp_statics} presents the numerical solution for quadratic preferences and uniform prior, with bliss point $L = -0.15$ for left voters and successive bliss points $R_1 = 0.15$, $R_2 = 0.25$, $R_3 = 0.35$, and $R_4 = 0.50$ (top to bottom) for the right voters. The sets of winning policy outcomes (in black) are $[-0.15, 0.15]$, $[-0.25, 0.15]$, $[-0.35, 0.1436]$, and $[-0.4098, 0.1098]$, respectively.} From \cref{lemma:more-extreme-more-persuadable}, as right voters become more extreme, their largest asymmetric interval of approved policy outcomes expands. This has two consequences. First, these voters are now more persuadable, which means that the challenger's odds of winning go up. Second, more extreme right voters approve policy outcomes further to the left. As a result, the left boundary of the equilibrium set of winning policy outcomes shifts to the left, as well. Interestingly, the right endpoint of the equilibrium set of winning policy outcomes, which is determined by left voters, may strictly decrease, too. This happens when right voters are or become persuadable by policy outcomes left of $\floor{A_L}$ (e.g., a change from $R_2$ to $R_3$/$R_4$, or from $R_3$ to $R_4$ in \cref{fig:comp_statics}).

\section{Discussion and Conclusion}\label{section:discussion}

This paper studied how a challenger advertises his privately known policy outcome to an electorate of voters with single-peaked and single-crossing preferences, using verifiable messages.
Below I discuss how the analysis of the main model can be extended to other assumptions common in the political economy literature. I consider the following extensions: the presence of a strategic incumbent, a citizen-candidate challenger, probabilistic voting, instrumental voting, and information spillovers.

\subsection*{Strategic Incumbent}
In the model, the incumbent does not advertise, and his policy outcome is known and normalized to zero. This assumption can be relaxed in a number of ways.

First, suppose that the status quo is a lottery, $\nu_0 \in \Delta X$ (independent of $\mu_0$), and the incumbent is still nonstrategic.
Then, voter $v$'s expected payoff from rejection is $\int u_v(y) d\nu_0(y)$.
While there are no left and right voters anymore (they were defined relative to $0$), voters $v$ and $w$ can still be defined as having {\it diametrically opposing preferences} if $\int u_v(x) d\mu(x) \geq \int u_v(y) d \nu_0(y)$ implies $\int u_w(x) d\mu(x) < \int u_w(y) d \nu_0(y)$ for all $\mu \in \Delta X$, meaning that at most one of these voters prefers to approve given the choice between $\nu_0$ and any $\mu$.
Then, an election is unwinnable for the challenger without targeted advertising if every decisive coalition requires convincing such voters.
With targeted advertising, the challenger induces different posteriors among different voters and is still able to convince voters with diametrically opposing preferences with a positive probability.

Next, suppose that the incumbent is strategic and can change $\nu_0$ to a common belief $\nu$ about the status quo policy outcome, perhaps by publicly advertising it. Assuming that the challenger has time to react, he still benefits from targeted advertising for the same reason as in the above paragraph. In fact, even if the incumbent chooses the status quo policy outcome, the challenger can win as long as not every decisive coalition includes a status quo voter.

Finally, if the candidates are symmetric (e.g., they both use targeted advertising to advertise their own and/or their opponent's policy outcome), then full unraveling of information takes place (\citenop{Janssen2017}; \citenop{schipper2019}). Therefore, the key insight of this paper is that having access to better targeted-advertising technology and/or better voter data allows politicians to win otherwise unwinnable elections. Without this advantage, voters choose the same candidate as they would under complete information.

\subsection*{Citizen-Candidate Challenger}

In the main model, the challenger is office-motivated: his payoff is $1$ if he wins and $0$ if he loses. Let us consider an extension in which the challenger is a citizen-candidate whose payoff is $u_x(x)$ if he wins and $u_x(0)$ if he loses.\footnote{I thank an anonymous referee for suggesting this extension. One possible microfoundation is as follows. The challenger's policy outcome $x$, which is his private information, is his {\it ideology}. The challenger's ideology is drawn from a common prior $\mu_0$. The challenger's message is a verifiable statement about his ideology. If he wins, he implements the policy outcome equal to his ideology, so his payoff is $u_x(x)$. If he loses, the status quo policy outcome $0$ persists, so his payoff is $u_x(0)$.}

It is easy to see that the equilibria described in Theorems \ref{thm1:PA} and \ref{thm2:TA} remain equilibria of the citizen-candidate games with public advertising and targeted advertising, respectively. In these equilibria, the challenger does not have profitable deviations {\it when he wins} because he reaches the highest possible payoff. On the other hand, he does not have profitable deviations {\it when he loses} due to voters' skeptical off-path beliefs, which ensure that he receives at most his complete information payoff if he deviates. Thus, as long as the challenger loses under complete information, whether his payoff depends on $x$ is irrelevant for the purpose of equilibrium construction.

Although the equilibria described in Propositions \ref{prop:baseline-PA} and \ref{prop:baseline-TA} maximize the challenger's {\it odds of winning} in a baseline election, they may not maximize his citizen-candidate ex-ante utility. It may be ex-ante optimal to let the set of winning policy outcomes include the most extreme challengers, as they would obtain the lowest payoff from losing. Characterizing the citizen-candidate's ex-ante optimal set of winning policy outcomes remains an open question.

\subsection*{Probabilistic Voting}

In the model, the challenger knows the set $\mathcal{D}$ of decisive coalitions. In this extension, I explore what happens if the challenger experiences uncertainty about which coalition is decisive. For concreteness, suppose that there are two possible voter bliss points, $L<0$ and $R>0$, and the unique minimal decisive coalition is $\{ L \}$, $\{ R \}$, or $\{ L,R \}$ with probability $\gamma_L$, $\gamma_R$, or $\gamma_{LR}$, respectively. The timing of the game is the same, except the uncertainty about which coalition is decisive is resolved after the challenger chooses his messages.

It is useful to establish the challenger's payoff from fully revealing his policy outcome, as that provides the lower bound on his equilibrium payoff. Under full revelation (i.e., if the challenger's strategy is to send message $\{x\}$ for all $x \in X$), his expected payoff is $\gamma_L$ if $x \in [ \floor{A_L},0 )$, $1$ if $x = 0$, and $\gamma_R$ if $x \in (0,\ceil{A_R}]$. Furthermore, much like in the main model, deviations to full revelation are the only ones we need to rule out; other deviations can be made unprofitable by imposing voters' skeptical off-path beliefs. Below we analyze pure-strategy equilibria of public and targeted advertising games.

\paragraph{Public Advertising}

For any pure-strategy equilibrium, consider the mapping from the challenger's realized policy outcome $x \in X$ to the actions taken by voters $L$ and $R$. Since these voters never approve under a common belief, there are three possible outcomes for each $x$: $L$ approves, $R$ approves, or no one approves. We can thus characterize a pure-strategy equilibrium outcome as a partition of the space of policy outcomes into those approved by the left coalition ($W_L$), those approved by the right coalition ($W_R$), and those approved by nobody ($W_\varnothing = X \cap W_L^c \cap W_R^c$). Furthermore, we can implement any such partition $(W_L,W_R,W_\varnothing)$ by letting the challenger reveal which partition element his policy outcome belongs to: his {\it partitional strategy} is to send message $W_L$ if $x \in W_L$, message $W_R$ if $x \in W_R$, and message $W_\varnothing$ otherwise. For that to be an equilibrium strategy, $W_v$ must be obedient for voter $v \in \{ L,R \}$. Moreover, the challenger must be obtaining at least his full revelation payoff.

One example of an equilibrium partitional strategy is ``divide-and-conquer''---getting left (right) voters to approve the largest interval of left (right) policy outcomes (by letting $W_L = I_L(-1,0)$ and $W_R = I_R(0,1)$). Another example of a partitional strategy is to maximize the odds of convincing the left voters (by letting $W_L = I_L(-1,1)$ and $W_R$ be the largest subset of $X \cap W_L^c$ obedient for $R$). In the end, the partitional strategy that maximizes the challenger's odds of winning among pure-strategy equilibria will depend on parameters $\gamma_L$ and $\gamma_R$.

\paragraph{Targeted Advertising} It is easy to see that any direct strategy $\sigma_{(M_L,M_R)}$, where $M_v$ is obedient for voter $v \in \{ L,R \}$ and includes her approval set (e.g., the strategy described in \cref{prop:baseline-TA}), is an equilibrium strategy. The challenger's ex-ante utility is then $\gamma_L \cdot \mu_0(M_L) + \gamma_R \cdot \mu_0(M_R) + \gamma_{LR} \cdot \mu_0( M_L \cap M_R )$. It is also easy to see that, depending on the values of $\gamma_L$, $\gamma_R$ and $\gamma_{LR}$, letting $(M_L,M_R) = (M_L^*,M_R^*)$---that is, maximizing the odds of convincing the mixed decisive coalition---may not maximize the challenger's ex-ante utility. For example, if the left coalition is almost certainly decisive ($\gamma_L \to 1$), then tailoring the message to that coalition by letting $M_L = I_L(-1,1)$ increases the odds of winning.

Overall, when there is uncertainty about decisive coalitions, the challenger faces tradeoffs that are not present in the main model---in particular, he may want to cater to coalitions that are more likely to be decisive. The challenger-preferred equilibrium outcome characterization, whether the focus on pure-strategy equilibria is without loss, and whether the challenger reaches the same payoff as in information design, remain open questions.

\subsection*{Instrumental Voting}

In the model, the voters have expressive preferences: if voter $v$ votes for policy outcome $y \in X$, then her payoff is $u_v(y)$. While this assumption is reasonable in large elections wherein the probability that an individual vote is pivotal is vanishingly small, in other elections voters may have instrumental concerns and derive utility also from the {\it winning} policy outcome. In this extension, I suppose that voters have both expressive (with probability $\beta \in [0,1]$) and instrumental (with probability $1-\beta$) concerns. Suppose that voter $v$'s payoff when she votes for policy outcome $y$ and the winning policy outcome is $y_w$ is $\widetilde{u}_v(y,y_w) = \beta u_v(y) + (1-\beta) u_v(y_w)$. Below I argue that polarized elections remain winnable if and only if $\beta > 0$, that is, whenever the voters do not have {\it purely} instrumental concerns.

For concreteness, consider a polarized baseline election with voters $L < 0$ and $R > 0$. For the same reasons as in the main model, the challenger-preferred equilibrium here is characterized by a pair $( \widetilde{M}_L, \widetilde{M}_R )$ of the voters' sets of approved policy outcomes obtained by solving\footnote{In that equilibrium, the challenger uses a direct strategy $\sigma_{(\widetilde{M}_L,\widetilde{M}_R)}$, voter $v \in \{ L,R \}$ approves after $\widetilde{M}_v$, and voters have skeptical off-path beliefs. That equilibrium maximizes the challenger's odds of winning across all equilibria because the pair $( \widetilde{M}_L, \widetilde{M}_R )$ also characterizes an optimal experiment that solves the information design problem. The obedience constraint of voter $v$ is obtained by simplifying $\int_{M_L \cap M_R} u_v(x) d\mu_0(x) + \int_{M_v \cap M_w^c} (\beta u_v(x) + (1-\beta) u_v(0)) d\mu_0(x) \geq \int_{M_v} u_v(0) d\mu_0(x)$.}
    \begin{align*}
         & \max_{(M_L,M_R) \subseteq X^2} \mu_0(M_L \cap M_R) \quad \text{subject to}
        \\& \quad \int_{M_L \cap M_R} \alpha_v(x) \, d\mu_0(x) + \beta \int_{M_v \cap M_w^c} \alpha_v(x) d\mu_0(x) \geq 0 \quad \text{for all } v \neq w \in \{L, R\}.
    \end{align*}
    Observe that the voters' obedience constraints now have two terms. The first term reflects the pivotality event, in which the winning policy outcome is $x$ because both voters approve. The second term is present only if the voters have expressive concerns (i.e., if $\beta > 0$), as in that case, the winning policy outcome remains $0$.

    Suppose first that voters are purely instrumental, that is, $\beta = 0$. Then the obedience constraint of voter $v$ in the challenger-preferred equilibrium becomes $\int_{\widetilde{M}_L \cap \widetilde{M}_R} \alpha_v(x) d\mu_0(x) \geq 0$.
    Crucially, that means that each voter's net payoff from approval is non-negative when she is pivotal, i.e., when both voters are recommended to approve. However, both voters are pivotal in the same event, whenever $x \in \widetilde{M}_L \cap \widetilde{M}_R$, and have a common posterior, $\mu_0(\cdot \sep \widetilde{M}_L \cap \widetilde{M}_R)$, in that event. By \eqref{SC1:incompatible}, the unique common belief under which left and right voters prefer to approve is $\delta_0$, so $\widetilde{M}_L \cap \widetilde{M}_R = \{0\}$. Therefore, the challenger's odds of convincing purely instrumental and jointly pivotal voters $L<0$ and $R>0$ are zero even in his most-preferred equilibrium of the targeted advertising game.

    If voters are not purely instrumental, that is, if $\beta > 0$, then there exists a pair $(M_L,M_R)$ that satisfies both voters' obedience constraints and $\mu_0(M_L \cap M_R) > 0$.\footnote{For example, it is easy to see that for any $\beta > 0$ there exists a small enough $\varepsilon > 0$ such that the set $M_v = A_v \cup [-\varepsilon,\varepsilon]$ is obedient for each $v \in \{ L,R \}$.} I thus conclude that polarized elections are winnable with targeted advertising even if there are only a small number of voters, as long as the voters' concerns are not purely instrumental.

    \subsection*{Information Spillovers}
    The final key assumption of the model is that there are no information spillovers, meaning that the challenger's targeted ads stay private. If left and right voters observed each other's messages, they would learn the same information, making targeted advertising as good as public disclosure. Therefore, informing voters of all ads transmitted during an electoral campaign is a useful tool to mitigate the negative effects of targeted advertising. In fact, 1,433 targeted ads of the Vote Leave campaign were released in the aftermath of the 2016 Brexit referendum, but the release occurred after the vote was finalized.\footnote{In July 2018, Facebook released the ads to the United Kingdom's Department for Culture, Media and Sport Committee as part of the committee's inquiry into fake news.}

    \begin{singlespace}
        \setlength\bibitemsep{5pt}
        \printbibliography
    \end{singlespace}

    \appendix

    \invisiblesection{}
    \section*{Appendix}\label{appendix}

    \begin{lemma}\label{lemma:SC2-appendix}
        Let \( W \subseteq X \) and suppose that voter $w \in V$ is more extreme than $v \in V$. Then:
        \begin{align*}
            \int_{W \cup A_v} \alpha_v(x) \, d\mu_0(x) \geq 0
            \quad \implies \quad
            \int_{W \cup A_w} \alpha_w(x) \, d\mu_0(x) \geq 0.
        \end{align*}
        ~\\[-2\baselineskip]
    \end{lemma}
    \begin{proof}
        By \eqref{SC2:more-extreme}, $\alpha_v(x) \geq 0 \implies \alpha_w(x) \geq 0$ for all $x \in X \smallsetminus \{ 0 \}$, so $A_v \subseteq A_w$. Next, write $W \cup A_w$ as a partition into disjoint sets $W \cup A_v$ and $A_w \cap W^c \cap A_v^c$. Therefore,
        \[
            \int_{W \cup A_w} \alpha_w d\mu_0 =
            \int_{W \cup A_v} \alpha_w d\mu_0 + \int_{A_w \cap W^c \cap A_v^c} \alpha_w d\mu_0,
        \]
        where the last term is non-negative because $\alpha_w(x) \geq 0$ for all $x \in A_w$. Consequently, $\int_{W \cup A_v} \alpha_v d\mu_0 \geq 0$ implies $\int_{W \cup A_v} \alpha_w d\mu_0 \geq 0$ (by \ref{SC2:more-extreme}), which implies $\int_{W \cup A_w} \alpha_w d\mu_0 \geq 0$.\qedhere
    \end{proof}

    \subsubsection*{ Proof of \cref{lemma:aux-solution} }

    If $\int\limits_l^r \alpha_v d\mu_0 \geq 0$, then $I_v(l,r) = [l,r]$ solves Problem \eqref{eqn:aux_prob}. We thus assume for the remainder of the proof that $\int\limits_l^r \alpha_v d\mu_0 < 0$.

    We first show that the set $I_v(l,r)$ is well-defined. Since $\alpha_v(x)$ is strictly quasiconcave, the set $S(d) := \{ x \in [l,r] \sep \alpha_v(x) \geq d \}$ is convex (i.e., is an empty set, a point, or a closed interval) for any $d \in \mathbbm{R}$. Furthermore, $S(d)$ expands as $d$ decreases,  i.e., if $d_2 < d_1$, then $S(d_1) \subseteq S(d_2)$; the inclusion is strict unless $S(d_1) = [l,r]$. Now, let $F(d) := \int_{ S(d) } \alpha_v d\mu_0$. By the dominated convergence theorem, $F(d)$ is continuous  ($\alpha_v(x)$ is continuous and bounded, while $\mu_0$ is a finite measure). Observe that $F(0) > 0$ since $S(0) = \{ x \in [l,r] \supseteq A_v \sep \alpha_v(x) \geq 0 \} = A_v$; here, $A_v \subseteq [l,r]$ because $l \leq \floor{A_v} < \ceil{A_v} \leq r$. Next, let $\underline{d} := \min_{x \in [l,r]} \alpha_v(x) = \min\{ \alpha_v(l), \alpha_v(r) \}$; note that $S(d) = [l,r]$ if and only if $d \leq \underline{d}$ and thus $F(\underline{d}) = \int\limits_l^r \alpha_v (x) d\mu_0(x) < 0$. Observe that $F(d)$ is strictly increasing in $d$ for $d \in [\underline{d},0]$. Indeed, if $\underline{d} \leq d_2 < d_1 \leq 0$, then
    \begin{align*}
        F(d_2) - F(d_1) = \int_{ S(d_2) \smallsetminus S(d_1) } \alpha_v d\mu_0 < 0,
    \end{align*}
    since $\alpha_v(x) < 0$ for all $x \in S(d_2) \smallsetminus S(d_1)$ and $S(d_1) \subset S(d_2)$ if $\underline{d} \leq d_2 < d_1$. By the intermediate value theorem, there exists a unique $d^* \in (\underline{d},0)$ such that $F(d^*) = 0$.

    We now show that all solutions to Problem~\eqref{eqn:aux_prob} agree with \( I^* \coloneqq S(d^*) = I_v(l,r) \) \(\mu_0\)-almost everywhere. Let \( \widetilde{I} \subseteq [l, r] \) be an arbitrary solution. The obedience constraint binds for $I^*$ and holds for $\widetilde{I}$, therefore:
    \begin{align*}
        \int_{I^* \cap \widetilde{I}} \alpha_v d\mu_0 + \int_{I^* \cap \widetilde{I}^c} \alpha_v d\mu_0 = 0
        \quad                                             & \text{ and } \quad
        \int_{ \widetilde{I} \cap I^*} \alpha_v d\mu_0 + \int_{\widetilde{I} \cap (I^*)^c} \alpha_v d\mu_0 \geq 0
        \\ \implies
        \int_{\widetilde{I} \cap (I^*)^c} \alpha_v d\mu_0 & \geq \int_{I^* \cap \widetilde{I}^c} \alpha_v d\mu_0.
    \end{align*}
    Now, since $I^* = S(d^*)$, we have $\alpha_v (x_1) < d^* \leq \alpha_v(x_2)$ for all $x_1 \in (I^*)^c$ and $x_2 \in I^*$. We thus obtain:
    \begin{align*}
        d^* \cdot \mu_0( \widetilde{I} \cap (I^*)^c) & \geq d^* \cdot \mu_0( I^* \cap \widetilde{I}^c )
        \\ \implies
        \mu_0( \widetilde{I} \cap (I^*)^c)           & \leq  \mu_0( I^* \cap \widetilde{I}^c )
        \\ \implies
        \mu_0( \widetilde{I} )                       & \leq  \mu_0( I^* ),
    \end{align*}
    since $d^* < 0$; these inequalities are strict unless $\mu_0( \widetilde{I} \cap (I^*)^c ) =\mu_0( I^* \cap \widetilde{I}^c ) = 0$. Therefore, $\widetilde{I}$ is a solution if and only if it agrees with $I^*$ $\mu_0$-almost everywhere.

    \begin{lemma}\label{lemma:IL-appendix}
        Let $I_L(l,1) = [a^*,b^*]$ be the interval solution to Problem \eqref{eqn:aux_prob} for voter $L<0$ with $l \in [-1,\floor{A_L}]$ and $r = 1$. Suppose that $\int_l^1 \alpha_L d\mu_0 < 0$. Let $S(d) := \{ x \in [l,1] \sep \alpha_L(x) \geq d \}$ and $d^* \in ( \min\{ \alpha_L(l), \alpha_L(1) \},0 )$ be the unique solution to $\int_{S(d^*)} \alpha_L(x) d\mu_0 = 0$. Then:
        \begin{enumerate}
            \item $a^* \in (l,\floor{A_L})$ and $b^* \in (0,1)$ if $d^* > \max\{ \alpha_L(l),\alpha_L(1) \}$;
            \item $a^* \in (l,\floor{A_L})$ and $b^* = 1$ if $\alpha_L(l) < d^* \leq \alpha_L(1)$;
            \item $a^* = l$ and $b^* \in (0,1)$ if $\alpha_L(l) \geq d^* > \alpha_L(1)$.
        \end{enumerate}
    \end{lemma}
    \begin{proof}
        Note that we showed the existence of $d^* \in ( \min\{ \alpha_L(l), \alpha_L(1) \},0 )$ in the proof of \cref{lemma:aux-solution}. Also note that $[a^*,b^*] = S(d^*)$. Furthermore, recall that $\alpha_L( \floor{A_L} ) = 0$; since $d^* < 0$, it follows that $d^* < \alpha_L( \floor{A_L} )$. Consequently, if $\alpha_L(l) < d^*$ (as in cases 1 and 2), then $\alpha_L(l) < \alpha_L( \floor{A_L} )$, $l < \floor{A_L}$ (since $\alpha_L$ is strictly increasing on $[-1,\floor{A_L}]$) and thus $a^* \in (l, \floor{A_L})$ from cases 1 and 2 is well-defined. We have:
        \begin{enumerate}
            \item If $d^* > \max\{ \alpha_L(l),\alpha_L(1) \}$, let $a^* \in (l,\floor{A_L})$ and $b^* \in (0,1)$ solve $\alpha_L(a^*) = \alpha_L(b^*) = d^*$.\footnote{Such $a^*$ and $b^*$ exist by the intermediate value theorem---for instance, for $a^*$, we have that $\alpha_L(x)$ is continuous and strictly increasing on $[l,\floor{A_L}]$, and $\alpha_L(l) < d^* < 0 = \alpha_L(\floor{A_L})$}
                  Then, $S(d^*) = [a^*,b^*]$ as desired.
            \item If $\alpha_L(l) < d^* \leq \alpha_L(1)$, let $a^* \in (l,\floor{A_L})$ solve $\alpha_L(a^*) = d^*$ so that $S(d^*) = [a^*,1]$.
            \item If $\alpha_L(l) \geq d^* > \alpha_L(1)$, let $b^* \in (0,1)$ solve $\alpha_L(b^*) = d^*$ so that $S(d^*) = [l,b^*]$.\qedhere
        \end{enumerate}
    \end{proof}

    \begin{lemma}\label{lemma:Ls-constraint-appendix}
        Let $-1 \leq l < l' \leq \floor{A_L}$. Also, let $I_L(l,1) = [a^*,b^*]$ and $I_L(l',1) = [a',b']$ be the interval solutions to Problem \eqref{eqn:aux_prob} for voter $L<0$ with $l,l'$ (resp.) and $r=1$. Then, $a^* \leq a'$ and $b^* \leq b'$.
    \end{lemma}
    \begin{proof}
        Observe that $-1 \leq l < l' \leq \floor{A_L}$ implies $\int_{l}^1 \alpha_L d\mu_0  < \int_{l'}^1 \alpha_L d\mu_0$. Therefore, if $\int_l^1 \alpha_L d\mu_0 \geq 0$, then $[a^*,b^*] = [l,1]$ and $[a',b'] = [l',1]$, and the claim is true. For the rest of the proof, assume $\int_l^1 \alpha_L d\mu_0 < 0$. Let $S(d,l) := \{ x \in [l,1] \sep \alpha_L(x) \geq d \}$ and $d^* \in ( \min\{ \alpha_L(l), \alpha_L(1) \},0 )$ be the unique solution to $\int_{S(d^*,l)} \alpha_L(x) d\mu_0 = 0$. Also, let $d' \in ( \min\{ \alpha_L(l'), \alpha_L(1) \},0 )$ be the unique solution to $\int_{S(d',l')} \alpha_L(x) d\mu_0 = 0$ if it exists (i.e., if $\int_{l'}^1 \alpha_L d\mu_0 < 0$). From \cref{lemma:IL-appendix}, we have three possible cases:
        \begin{enumerate}
            \item $d^* > \max\{ \alpha_L(l),\alpha_L(1) \}$ so that $a^* \in (l,\floor{A_L})$, $b^* \in (0,1)$ and $\alpha_L(a^*) = \alpha_L(b^*) = d^*$:
                  \begin{itemize}
                      \item[(i)] $\alpha_L(l') \leq d^* = \alpha_L(a^*) \iff l' \leq a^*$, so $S(d^*,l) \subseteq [l',1]$, $S(d',l') = S(d^*,l)$, and thus $[a',b'] = [a^*,b^*]$.
                      \item[(ii)] $d^* = \alpha_L(a^*) < \alpha_L(l') \iff a^* < l'$. If $\int_{l'}^1 \alpha_L d\mu_0 \geq 0$, then $[a',b'] = [l',1]$ and the claim holds since $a^* < l'$ and $b^* < 1$. Else, if $\int_{l'}^1 \alpha_L d\mu_0 < 0$, then $d' \in ( \min\{ \alpha_L(l'),\alpha_L(1) \},0 )$ exists, and:
                            \[
                                \int_{S(d',l')} \alpha_L d\mu_0 =0 = \int_{S(d^*,l)}\alpha_L d\mu_0 = \int_{a^*}^{b^*}\alpha_L d\mu_0 < \int_{l'}^{b^*} \alpha_Ld\mu_0 = \int_{S(d^*,l')} \alpha_L d\mu_0,
                            \]
                            so that $d' < d^*$
                            since the function $\int_{S(d,l')}\alpha_L d\mu_0$ is strictly increasing in $d$ (see the proof of \cref{lemma:aux-solution}). Therefore, we have
                            \[
                                \min\{ \alpha_L(l'),\alpha_L(1) \} = \alpha_L(1) < d' < d^* = \alpha_L(a^*) < \alpha_L(l'),
                            \]
                            and by \cref{lemma:IL-appendix} (case 3), $[a',b'] = [l',b']$, where $\alpha_L(b') = d' < d^* = \alpha_L(b^*) \iff b' > b^*$, and the claim holds.
                  \end{itemize}
            \item $\alpha_L(l) < d^* \leq \alpha_L(1)$ so that $a^* \in (l,\floor{A_L})$, $b^* = 1$ and $\int_{a^*}^1 \alpha_L d\mu_0 = 0$. Then, if $a^* \leq l'$, we have $\int_{l'}^1 \alpha_L d\mu_0 > 0$ so that $[a',b'] = [l',1]$. If, on the other hand, $l < l' < a^*$, then $\alpha_L(l') < \alpha_L(a^*) = d^*$, so $S(d^*,l) \subseteq [l',1]$, $S(d',l') = S(d^*,l)$, and thus $[a',b'] = [a^*,b^*]$. Overall, in case 2 we have $[a',b'] = [ \max\{ a^*,l'  \},1]$ and the claim holds.
            \item If $\alpha_L(l) \geq d^* > \alpha_L(1)$, then $a^* = l$, $b^* \in (0,1)$, $\alpha_L(b^*) = d^*$ and $\int_{l}^{b^*} \alpha_L d\mu_0 = 0$. If $\int_{l'}^1 \alpha_L d\mu_0 \geq 0$, then $[a',b'] = [l',1]$ and the claim holds. If $\int_{l'}^1 \alpha_L d\mu_0 < 0$, then $\alpha_L(l') > \alpha_L(l) \geq d^* > d' > \alpha_L(1)$ by the same argument as in case 1.(ii). Consequently, by \cref{lemma:IL-appendix} (case 3), $[a',b'] = [l',b']$, where $\alpha_L(b') = d' < d^* = \alpha_L(b^*) \iff b' > b^*$. Overall, in case 3 we have $[a',b'] = [l',b']$, where $l' > l$, $b' > b$, and the claim holds.\qedhere
        \end{enumerate}
    \end{proof}

    \subsubsection*{Proof of \cref{prop:baseline-TA}}

    \noindent\ul{\bf Step 1}: Show that Problem \eqref{problem:AUX-TA}
    admits a solution
    \begin{enumerate}
        \item $M_L^* = [ \floor{A_L},b_L ]$ and $M_R^* = [ a_R, \ceil{A_R} ]$ if $\floor{A_L} \leq a_R$ and $b_L \leq \ceil{A_R}$;
        \item $M_L^* = I_L( a_R,1 )$ and $M_R^* = [ a_R, \ceil{A_R} ]$ if $a_R < \floor{A_L}$ and $b_L \leq \ceil{A_R}$;
        \item $M_L^* = [ \floor{A_L},b_L ]$ and $M_R^* = I_R( -1,b_L )$ if $\floor{A_L} \leq a_R$ and $\ceil{A_R} < b_L$.
    \end{enumerate}

    Suppose that $(\widetilde{M}_L,\widetilde{M}_R)$ such that $A_v \subseteq \widetilde{M}_v$ for each $v$ is an arbitrary solution to Problem \eqref{problem:AUX-TA}.\footnote{If a pair $(M_L,M_R)$ is a solution to \eqref{problem:AUX-TA}, then the pair $(M_L \cup A_L,M_R \cup A_R)$ is also a solution. We consider solutions that include voters' approval sets for ease of exposition.} Let $\widetilde{W} := \widetilde{M}_L \cap \widetilde{M}_R$, $W^* := M_L^* \cap M_R^*$, and:
    \begin{align*}
         & Z_L := [-1,0] \cap W^* \cap \widetilde{W}^c, \quad   & Z_R := [0, 1] \cap W^* \cap \widetilde{W}^c,  \\
         & Y_L := [-1,0] \cap \widetilde{W} \cap (W^*)^c, \quad & Y_R := [0,1] \cap \widetilde{W} \cap (W^*)^c.
    \end{align*}
    We will use the voters' obedience constraints to show that $\mu_0(W^*) \geq \mu_0(\widetilde{W})$, which implies that $(M_L^*,M_R^*)$ is a solution to Problem \eqref{problem:AUX-TA}. By contradiction, suppose that $ \mu_0(\widetilde{W}) > \mu_0(W^*) \iff \mu_0(Y_L) + \mu_0(Y_R) > \mu_0(Z_L) + \mu_0(Z_R)$. We will address cases 1 and 2 (each with multiple subcases), in which $M_R^* = [a_R, \lceil A_R \rceil]$; case 3 is proved analogously to case 2.
    \begin{itemize}
        \item[1.1] $-1 < \floor{A_L} \leq a_R$ and $b_L \leq \ceil{A_R} < 1$. In this case, $M_L^* = I_L(\floor{A_L},1) = [ \floor{A_L},b_L ]$; also, $Z_L = [a_R,0] \cap \widetilde{W}^c$, $Y_L = [-1,a_R) \cap \widetilde{W}$, $Z_R = [0,b_L] \cap \widetilde{W}^c$, and $Y_R = (b_L,1] \cap \widetilde{W}$.

        First, note that $M_L^*$ binds $L$'s obedience constraint since $b_L < 1$, we have $\int_{\floor{A_L}}^{1} \alpha_L d\mu_0 < 0$ and $\int_{\floor{A_L}}^{b_L} \alpha_L d\mu_0 = 0$. Similarly, $M_R^*$ binds $R$'s obedience constraint since $-1 < a_R$.

              Then, we can partition\footnote{The objects that we call partitions in this proof are partitions $\mu_0$-almost everywhere, i.e., there could be points belonging to multiple elements of the partition.} $M_L^* = [ \floor{A_L},b_L ]$ into three sets $[\floor{A_L},0] = A_L$, $[0,b_L] \cap \widetilde{W}$ and $[0,b_L] \cap \widetilde{W}^c = Z_R$, and write down $L$'s (binding) obedience constraint for $M_L^*$ as:
              \begin{equation}\label{appendix-P2-step1-obedience-star}
                  \int_{A_L} \alpha_L d\mu_0 + \int_{[0,b_L] \cap \widetilde{W}} \alpha_L d\mu_0 + \int_{Z_R}\alpha_L d\mu_0 = 0.
              \end{equation}

              Similarly, we partition $\widetilde{M}_L$ into four sets $[-1,\floor{A_L})\cap \widetilde{M}_L$, $A_L$, $[0,b_L]\cap \widetilde{M}_L$ and $(b_L,1] \cap \widetilde{M}_L$, and write down $L$'s obedience constraint for $\widetilde{M}_L$ as:
              \begin{equation}\label{appendix:step1-1-tilde}
                  \begin{aligned}
                      0
                       & \leq
                      \int_{[-1,\floor{A_L})\cap \widetilde{M}_L} \alpha_L d\mu_0 +
                      \int_{A_L} \alpha_L d\mu_0 +
                      \int_{[0,b_L]\cap \widetilde{M}_L} \alpha_L d\mu_0 +
                      \int_{(b_L,1] \cap \widetilde{M}_L} \alpha_L d\mu_0 \\
                       & \leq
                      \int_{[-1,\floor{A_L})\cap \widetilde{W}} \alpha_L d\mu_0 +
                      \int_{A_L} \alpha_L d\mu_0 +
                      \int_{[0,b_L]\cap \widetilde{W}} \alpha_L d\mu_0 +
                      \int_{Y_R} \alpha_L d\mu_0,
                  \end{aligned}
              \end{equation}
              where the last inequality holds because $\widetilde{W} \subseteq \widetilde{M}_L$ (so that $[0,b_L]\cap \widetilde{W} \subseteq [0,b_L]\cap \widetilde{M}_L$ and $Y_R \subseteq (b_L,1] \cap \widetilde{M}_L$) and $\alpha_L$ is negative outside $A_L$. Combining \eqref{appendix-P2-step1-obedience-star} and \eqref{appendix:step1-1-tilde}, we get:
              \[
                  \int_{Y_R}\alpha_L d\mu_0 - \int_{Z_R}\alpha_L d\mu_0 \geq - \int_{[-1,\floor{A_L})\cap \widetilde{W}} \alpha_L d\mu_0 \geq 0.
              \]
              Since $Z_R \subseteq [0,b_L]$, $Y_R \subseteq (b_L,1]$, and $\alpha_L$ is strictly decreasing on $[0,1]$, we get:
              \[
                  \alpha_L(b_L) \mu_0(Y_R) \geq \int_{Y_R}\alpha_L d\mu_0
                  \geq \int_{Z_R}\alpha_L d\mu_0 \geq \alpha_L(b_L) \mu_0(Z_R),
              \]
              which implies that $\mu_0(Z_R) \geq \mu_0(Y_R)$ since $\alpha_L(b_L) < 0$. Using the same argument (i.e., by comparing the terms in $R$'s obedience constraints for $M_R^*$ and $\widetilde{M}_R$), we obtain $\mu_0(Z_L) \geq \mu_0(Y_L)$. Therefore, $\mu_0(Z_L) + \mu_0(Z_R) \geq \mu_0(Y_L) + \mu_0(Y_R)$, a contradiction.

        \item[1.2] $-1 = \floor{A_L} = a_R$ and $b_L \leq \ceil{A_R} < 1$. In this case, $L$'s constraint for $M_L^*$ is binding (since $b_L < 1$), while $R$'s constraint for $M_R^*$ may or may not bind.

              Using the fact that $M_L^*$ binds $L$'s constraint, we obtain $\mu_0(Z_R) \geq \mu_0(Y_R)$ (see step 1.1). While $R$'s constraint for $M_R^*$ may not bind, from $-1 = \floor{A_L} = a_R$ we have $M_R^* = [-1,\ceil{A_R}]$ and $[-1,0] \cap W^* = [-1,0]$, so $Y_L = \varnothing$ and $\mu_0 (Z_L) \geq 0 = \mu_0(Y_L)$. Therefore, $\mu_0(Z_L) + \mu_0(Z_R) \geq \mu_0(Y_L) + \mu_0(Y_R)$, a contradiction.

        \item[1.3] $-1 < \floor{A_L} \leq a_R$ and $b_L = \ceil{A_R} = 1$. This case is analogous to 1.2.
        \item[1.4] $-1 = \floor{A_L} = a_R$ and $b_L = \ceil{A_R} = 1$. This case is impossible because then $\int_{-1}^{1} \alpha_v d\mu_0 \geq 0 $ for each $v \in \{ L,R \}$, i.e., both voters weakly prefer to approve under the prior $\mu_0$; that contradicts \eqref{SC1:incompatible}.

        \item[2.1] $a_R \leq \floor{ I_L(-1,1) } < \floor{A_L}$ and $b_L \leq \ceil{A_R}$. In this case, $M_L^* = W^*= I_L(-1,1)$ and $M_R^* = I_R(-1,\ceil{A_R})$.\footnote{See the proof of \cref{lemma:Ls-constraint-appendix} with $l = -1$ and $l' = a_R$. If $a_R \leq \floor{I_L(-1,1)}$, then $I_L(a_R,1) = I_L(-1,1)$. If, on the other hand, $a_R > \floor{I_L(-1,1)}$, then $\floor{I_L(a_R,1)} = a_R$.} Since $W^*$ solves Problem \eqref{eqn:aux_prob} with $l=-1$ and $r = 1$ for voter $L$, which is Problem \eqref{problem:AUX-TA} without $R$'s obedience constraint, the pair $(M_L^*,M_R^*)$ also solves \eqref{problem:AUX-TA}.

        \item[2.2] $\floor{ I_L(-1,1) } < a_R < \floor{A_L}$ and $b_L \leq \ceil{A_R}$. In this case, $M_L^* = I_L(a_R,1) =: [a_R,b]$, $M_R^* = I_R(-1,\ceil{A_R}) = [a_R,\ceil{A_R}]$ and $W^* = [a_R,b]$. Also, $Z_L = [a_R,0]\cap \widetilde{W}^c$, $Y_L = [-1,a_R) \cap \widetilde{W}$, $Z_R = [0,b] \cap \widetilde{W}^c$, and $Y_R = (b,1] \cap \widetilde{W}$. It is worth mentioning that $b_L$ and $b$ are different objects: $b_L$ is the upper bound of $I_L(\floor{A_L},1) = [\floor{A_L},b_L]$, while $b$ is the upper bound of $I_L(a_R,1) = [a_R,b]$. By \cref{lemma:Ls-constraint-appendix}, we have $b \leq b_L$.
        
        Now, since $-1 \leq \floor{I_L(-1,1)} < a_R$, we have $\int_{-1}^{\ceil{A_R}}\alpha_R d\mu_0 < 0$, so $\int_{a_R}^{\ceil{A_R}}\alpha_R d\mu_0 = 0$ and $R$'s constraint for $M_R^*$ binds. From that, $\mu_0(Z_L) \geq \mu_0(Y_L)$ (see step 1.1).

              For voter $L$, two cases are possible: $b=1$ and $b < 1$. If $b = 1$, then $[a_R,1]$ satisfies $L$ and $R$'s obedience constraints, i.e., both voters prefer to approve under belief $\mu_0(\cdot \sep [a_R,1])$, which contradicts \eqref{SC1:incompatible}. Therefore, $b < 1$.

              Now, we have $\int_{a_R}^1 \alpha_L d\mu_0 < 0$ and $\int_{a_R}^{b} \alpha_L d\mu_0 =0$, i.e., $M_L^*$ binds $L$'s obedience constraint. Partition $M_L^*$ into $[a_R,\floor{A_L}) \cap \widetilde{W} =: i_1$, $[a_R,\floor{A_L}) \cap \widetilde{W}^c = Z_L \cap A_L^c$, $A_L$, $[0,b]\cap \widetilde{W} =: i_2$, $[0,b] \cap \widetilde{W}^c = Z_R$, and write down $L$'s (binding) obedience constraint for $M_L^*$ as:
              \begin{equation}\label{appendix-step2-MRstar}
                  \int_{i_1} \alpha_L d\mu_0 +
                  \int_{Z_L \cap A_L^c} \alpha_L d\mu_0 +
                  \int_{A_L} \alpha_L d\mu_0 +
                  \int_{i_2} \alpha_L d\mu_0 +
                  \int_{Z_R} \alpha_L d\mu_0
                  = 0.
              \end{equation}
              Similarly, partition $\widetilde{M}_L$ into $[-1,a_R) \cap \widetilde{M}_L$, $[a_R,\floor{A_L}) \cap \widetilde{M}_L$, $A_L$, $[0,b] \cap \widetilde{M}_L$, $(b,1] \cap \widetilde{M}_L$. Using the fact that $\widetilde{W} \subseteq \widetilde{M}_L$ and $\alpha_L$ is negative outside $A_L$, from the obedience constraint for $\widetilde{M}_L$ we obtain
              \begin{equation}\label{appendix-step2-MLtilde}
                  \int_{Y_L} \alpha_L d\mu_0 +
                  \int_{i_1} \alpha_L d\mu_0 +
                  \int_{A_L} \alpha_L d\mu_0 +
                  \int_{i_2} \alpha_L d\mu_0 +
                  \int_{Y_R} \alpha_L d\mu_0
                  \geq 0.
              \end{equation}
              Combining \eqref{appendix-step2-MRstar} and \eqref{appendix-step2-MLtilde}, we get
              \[
                  \int_{Y_L}\alpha_L d\mu_0 + \int_{Y_R}\alpha_L d\mu_0 \geq \int_{Z_L \cap A_L^c}\alpha_L d\mu_0 + \int_{Z_R}\alpha_L d\mu_0
              \]
              and, since $\alpha_L(y) \leq \alpha_L( a_R ) \leq \alpha_L(z) \leq 0$ for all $y \in Y_L$, $z \in Z_L \cap A_L^c$ and $\alpha_L(y) \leq \alpha_L(b) \leq \alpha_L(z) \leq 0$ for all $y \in Y_R$, $z \in Z_R$, we obtain
              \begin{align*}
                  \alpha_L( a_R ) \mu_0(Y_L) + \alpha_L(b) \mu_0(Y_R)
                  \geq & \alpha_L( a_R ) \mu_0(Z_L \cap A_L^c) + \alpha_L(b) \mu_0(Z_R) \\
                  =    & \alpha_L( a_R ) \mu_0(Z_L) + \alpha_L(b) \mu_0(Z_R)
                  - \alpha_L( a_R ) \mu_0(Z_L \cap A_L)                                 \\
                  \geq & \alpha_L( a_R ) \mu_0(Z_L) + \alpha_L(b) \mu_0(Z_R)
              \end{align*}
              since $\alpha_L(a_R) < 0$.

              Next, observe that if $b < 1$, then $\alpha_L(a_R) \geq \alpha_L(b)$. Indeed, since $[a_R,b] = I_L(a_R,1)$ and $\int_{a_R}^1 \alpha_L d\mu_0 < 0$, from \cref{lemma:aux-solution} we get that $[a_R,b] = \{ x \in [a_R,1] \sep \alpha_L(x) \geq d^* \}$ for some $d^* < 0$. In particular, we have $\alpha_L(x) \geq d^*$ for all $x \in [0,b]$ and $\alpha_L(x) < d^*$ for all $x \in (b,1]$, which is a non-empty set if $b < 1$. By the continuity and strict monotonicity of $\alpha_L$ on $[0,1]$, we have $\alpha_L(b) = d^*$ so that $\alpha_L(a_R) \geq d^* = \alpha_L(b)$.

              Now, dividing the last inequality by $\alpha_L(b) < 0$, we get:
              \[
                  \frac{\alpha_L(a_R)}{\alpha_L(b)} \mu_0(Y_L) + \mu_0(Y_R) \leq
                  \frac{\alpha_L(a_R)}{\alpha_L(b)} \mu_0(Z_L) + \mu_0(Z_R).
              \]
              Next, add $\left( 1-\frac{\alpha_L(a_R)}{\alpha_L(b)} \right)(\mu_0(Y_L) + \mu_0(Z_L))$ to both sides to obtain:
              \[
                  \mu_0(Y_L) + \mu_0(Y_R) + \left( 1-\frac{\alpha_L(a_R)}{\alpha_L(b)} \right) \mu_0(Z_L)
                  \leq
                  \mu_0(Z_L) + \mu_0(Z_R) + \left( 1-\frac{\alpha_L(a_R)}{\alpha_L(b)} \right) \mu_0(Y_L).
              \]
              Rearranging terms, we get:
              \begin{align*}
                  \mu_0(Z_L) + \mu_0(Z_R) \geq & \mu_0(Y_L) + \mu_0(Y_R) +
                  \left( 1-\frac{\alpha_L(a_R)}{\alpha_L(b)} \right)( \mu_0(Z_L) - \mu_0(Y_L) ) \\
                  \geq                         & \mu_0(Y_L) + \mu_0(Y_R)
              \end{align*}
              since $1-\frac{\alpha_L(a_R)}{\alpha_L(b)} \geq 0$ if $\alpha_L(a_R) \geq \alpha_L(b)$ and $\mu_0(Z_L) - \mu_0(Y_L) \geq 0$ from the previous calculations in the beginning of step 2.2. Therefore, $\mu_0(Z_L) + \mu_0(Z_R) \geq \mu_0(Y_L) + \mu_0(Y_R)$, a contradiction.
        \item[3] $\floor{A_L} \leq a_R$ and $\ceil{A_R} < b_L$. This case is analogous to case 2.
    \end{itemize}

    \noindent\ul{\bf Step 2}: Equilibrium Characterization.

    Observe that the mixed decisive coalition is $\{ L,R \}$; the set $M_L^* \cap M_R^* = [a,b]$ is an interval such that $a < 0$ and $b > 0$; the set $A_v \cup [a,b] = M_v^*$ satisfies the obedience constraint of voter $v \in \{ L,R \}$. Thus, $[a,b]$ is implementable and the equilibrium that implements is described in the proof of \cref{thm2:TA}; the equilibrium strategy of the challenger is $\sigma_{(M_L^*,M_R^*)}$

    \noindent\ul{\bf Step 3}: Show that there exists an optimal experiment that is characterized by $(M_L^*,M_R^*)$, for which the challenger's odds of winning are $\mu_0( M_L^* \cap M_R^* )$.

    The challenger's information design problem is formulated as follows.\footnote{In what follows, we employ the revelation principle (\citenop{Bergemann2016}) that allows us to restrict attention to action recommendations that are obeyed.} First, the challenger chooses and commits to an experiment, which is a measurable map $\psi: X \to \Delta \{0,1\}^2$. Next, the challenger's policy outcome $x$ is realized according to $\mu_0$ and the signals $s_L \in \{0,1\}$ and $s_R \in \{0,1\}$ are sent to voters with bliss points $L$ and $R$ (resp.) with probability $\psi( (s_L,s_R) \sep x )$. Then, voter $v \in \{ L,R \}$ privately observes her signal $s_v$, forms a posterior belief $\mu_v(\cdot \sep s_v) \in \Delta X$ using the Bayes rule, and approves after $s_v =1$ and rejects after $s_v = 0$. Let $\psi_v( s_v \sep x ) := \sum\limits_{s_{-v} \in\{0,1\} } \psi( (s_v,s_{-v}) \sep x )$ be the marginal probability that $v$ receives signal $s_v$. For $v$ to approve after signal $s_v = 1$, her net payoff from approval must be non-negative:
$$
    \int \alpha_v(x) d\mu_v(x \sep 1) \geq 0 \iff
    \int \alpha_v(x) \psi_v(1 \sep x) d\mu_0(x) \geq 0.
$$
Similarly, for $v$ to reject after signal $s_v = 0$, her expected net payoff from approval must be negative, $\int \alpha_v(x) \psi_v(0 \sep x) d\mu_0(x) < 0$. An optimal experiment maximizes the challenger's odds of winning and solves
\begin{align*}
     & \max\limits_{\psi}
    \int \psi( (1,1) \sep x )d\mu_0(x)
    \quad \text{subject to}
    \\
     & \int \alpha_v(x) \psi_v(1\sep x) d\mu_0(x) \geq 0
    \quad \text{and} \quad
    \int \alpha_v(x) \psi_v(0\sep x) d\mu_0(x) < 0,
    \quad \forall v \in \{L,R\}.
    \\[-2\baselineskip]
\end{align*}

Moreover, we can drop the less-than-zero constraints as letting $\psi_v(1 \sep x) = 1$ for each $v \in \{ L,R \}$ and $x \in A_v$ weakly increases the objective and loosens $v$'s constraints. Now, since each $\alpha_v$ is bounded, $\mu_0$ is a finite and atomless positive measure, and $X$ is a closed interval, a deterministic optimal experiment $\psi^*: X \to \{0,1\}^2$ exists (by an argument similar to one in the proof of Proposition 2 in \citenop{TitovaZhangPVI}). Next, let $C_v^* := \{ x \in X \sep \psi^*_v (1\sep x) = 1 \}$ for each $v \in \{L,R\}$ be the set of policy outcomes that $v$ is recommended to approve. Then, an optimal deterministic experiment  is characterized by a pair $(C_L^*,C_R^*)$ that solves Problem \eqref{problem:AUX-TA}. Hence, $\mu_0(M_L^* \cap M_R^*)$ is the challenger's odds of winning in the information design problem.

\subsubsection*{Proof of \cref{lemma:more-extreme-more-persuadable}}

We prove this statement for left voters $w < v < 0$. The proof for right voters $0 < v < w$ is analogous. From the definition of $I_v$ ($I_w$) as a solution to Problem \eqref{eqn:aux_prob} with $l = \floor{A_v}$ ($l = \floor{A_w}$) and $r = 1$, three cases are possible:

\paragraph{Case 1} $\int_{ \floor{A_v} }^1 \alpha_v d\mu_0 \geq 0$. Then, $I_v = [ \floor{A_v}, 1 ]$. By \eqref{SC2:more-extreme}, $\int_{ \floor{A_v} }^1 \alpha_v d\mu_0 \geq 0 \implies \int_{ \floor{A_v} }^1 \alpha_w d\mu_0 \geq 0 \implies \int_{ \floor{A_w} }^1 \alpha_w d\mu_0 \geq 0$ so $I_w = [ \floor{A_w}, 1 ]$ and $I_w \supseteq I_v$ since $\floor{A_w} \leq \floor{A_v}$.

\paragraph{Case 2} $\int_{ \floor{A_v} }^1 \alpha_v d\mu_0 < 0$ and $\int_{ \floor{A_w} }^1 \alpha_w d\mu_0 \geq 0$. Then, $I_v = [ \floor{A_v},b_v ]$, where $\int_{\floor{A_v}}^{b_v} \alpha_v d\mu_0 = 0$ and $b_v < 1$. Therefore, $I_v = [ \floor{A_v}, b_v ] \subset [ \floor{A_w}, 1 ] = I_w$.

\paragraph{Case 3} $\int_{ \floor{A_v} }^1 \alpha_v d\mu_0 < 0$ and $\int_{ \floor{A_w} }^1 \alpha_w d\mu_0 < 0$. Then, $I_v = [ \floor{A_v},b_v ]$, where $\int_{\floor{A_v}}^{b_v} \alpha_v d\mu_0 = 0$ and $b_v < 1$; the same is true for $I_w$---in particular, $\int_{\floor{A_w}}^{b_w} \alpha_w d\mu_0 = 0$. By \eqref{SC2:more-extreme}, $\int_{\floor{A_v}}^{b_v} \alpha_v d\mu_0 = 0 \implies \int_{\floor{A_v}}^{b_v} \alpha_w d\mu_0 \geq 0 \implies \int_{\floor{A_w}}^{b_v} \alpha_w d\mu_0 \geq 0$ (the last inequality is strict unless $\floor{A_v} = \floor{A_w}$). Now, since the function $\int_{\floor{A_w}}^{z} \alpha_w d\mu_0$ is continuous and strictly decreasing in $z$ for $z \in (0,1)$, we have
\[
    \int_{\floor{A_w}}^{b_v} \alpha_w d\mu_0 \geq (>)~0 = \int_{\floor{A_w}}^{b_w} \alpha_w d\mu_0 \iff b_v \leq (<)~b_w,
\]
so that $I_v = [ \floor{A_v}, b_v ] \subseteq [ \floor{A_w}, b_w ] = I_w$. Furthermore, unless $\floor{A_v} = \floor{A_w}$, we have $\floor{A_w} < \floor{A_v}$ and $b_v < b_w$. \qedhere

\subsubsection*{Proof of \cref{prop:comp_stats}}

Let $\{L,0,E\}$ be the baseline electorate with the more extreme right voter $E > R$. Let $(\widetilde{M}_L,\widetilde{M}_E)$ be the solution to \eqref{problem:AUX-TA} for this electorate (described in \cref{prop:baseline-TA}). Note that $b_L \leq \ceil{A_R}$ and $R < E$ imply that $b_L \leq \ceil{A_E}$, so $(M_L^*,M_R^*)$ and $(\widetilde{M}_L,\widetilde{M}_E)$ are both described by Case 1 or 2 of \cref{prop:baseline-TA}. Therefore, $M_R^* = [a_R, \ceil{A_R}]$, $\widetilde{M}_E = [a_E, \ceil{A_E}]$ and, by \cref{lemma:more-extreme-more-persuadable}, $a_E \leq a_R$. To simplify exposition, let $W^* := M_L^* \cap M_R^*$ and $\widetilde{W} := \widetilde{M}_L \cap \widetilde{M}_E$.

If $\floor{A_L} \leq a_E \leq a_R$, then both elections fall into case 1 of \cref{prop:baseline-TA}. We have $W^* = [a_R,b_L]$, $\widetilde{W} = [a_E,b_L]$, and the claims of the proposition are true because $a_E \leq a_R$.

Next, suppose that $a_E < \floor{A_L}$, in which case $\widetilde{M}_L = \widetilde{W} = I_L(a_E,1) =: [\widetilde{a}_L,\widetilde{b}_L]$. Observe that $\mu_0(I_L(l,1))$ is {\it decreasing} in $l$ (which is a parameter in Problem \ref{eqn:aux_prob}) on $[-1,\floor{A_L}]$: increasing $l$ shrinks the feasible region $[l,1]$, so the objective value $\mu_0(I_L(l,1))$ can only go down. Therefore, $\mu_0(I_L(a_E,1)) \geq \mu_0( I_L(l,1) )$ for all $l \in [a_E,\floor{A_L}]$. In particular, if $a_E < \floor{A_L} \leq a_R$, then $\mu_0(I_L(a_E,1)) \geq \mu_0( I_L(\floor{A_L},1) ) \geq \mu_0([a_R,b_L])$, and if $a_E \leq a_R < \floor{A_L}$, then $\mu_0(I_L(a_E,1)) \geq \mu_0( I_L(a_R,1) )$. Either way, $\mu_0(\widetilde{W}) \geq \mu_0(W^*)$.

Finally, from \cref{lemma:Ls-constraint-appendix}, both boundaries of $I_L(a_E,1)$ are left of the corresponding boundaries of $I_L(a_R,1)$ (often strictly so---see the proof of \cref{lemma:Ls-constraint-appendix}). If $a_E \leq a_R < \floor{A_L}$, then $\widetilde{W} = I_L(a_E,1)$ and $W^* = I_L(a_R,1)$ so $\widetilde{W}$ is left of $W^*$. If, on the other hand, $a_E < \floor{A_L} \leq a_R$, then $\widetilde{W} = I_L(a_E,1)$ is left of $I_L(\floor{A_L},1) = [\floor{A_L},b_L]$, which is a superset of $[a_R,b_L] = W^*$, so $\widetilde{W}$ is left of $W^*$. \qedhere

\end{document}

%% file: header.tex
\usepackage{amssymb,amsfonts,amsmath,amsthm,mathtools,bm}

\usepackage{bbm}

\usepackage[dvipsnames]{xcolor}

\definecolor{inlinemath}{rgb}{0.75, 0.25, 0.0}
\let\SetColor\color

\everymath\expandafter{\the\everymath\SetColor{inlinemath}}

\usepackage{etoolbox}
\makeatletter
\AtBeginEnvironment{align*}{\let\SetColor\@gobble \color{black}}
\AtBeginEnvironment{align}{\let\SetColor\@gobble \color{black}}
\AtBeginEnvironment{aligned*}{\let\SetColor\@gobble \color{black}}
\AtBeginEnvironment{aligned}{\let\SetColor\@gobble \color{black}}
\AtBeginEnvironment{equation*}{\let\SetColor\@gobble \color{black}}
\AtBeginEnvironment{equation}{\let\SetColor\@gobble \color{black}}
\AtBeginEnvironment{underline}{\let\SetColor\@gobble \color{black}}
\AtBeginEnvironment{tabular}{\let\SetColor\@gobble \color{black}}
\makeatother

\definecolor{MyBlue}{RGB}{0,91,148}
\definecolor{MyRed}{RGB}{200,15,62}
\definecolor{MyPurple}{rgb}{0.44, 0.16, 0.39}

\newcommand\lblue[1]{\textcolor{MyBlue}{#1}}
\newcommand\lred[1]{\textcolor{MyRed}{#1}}
\newcommand\lpurple[1]{\textcolor{MyPurple}{#1}}

\usepackage{subcaption,setspace,enumerate,color}
\usepackage{graphicx,multicol}
\usepackage[normalem]{ulem}
\usepackage{soul}
\usepackage{adjustbox}

\usepackage{tabu}
\usepackage{makecell,multirow}

\usepackage{titlesec}
\titleformat*{\section}{\sc \centering \Large}
\titleformat*{\subsection}{\sc \centering \large}
\titleformat*{\subsubsection}{\sc \large}
\titlelabel{\thetitle. }
\titlespacing*{\section}
{0pt}{1.5ex plus 0.3ex minus .1ex}{1.5ex plus .1ex}
\titlespacing*{\subsection}
{0pt}{1.5ex plus 0.3ex minus .1ex}{1.5ex plus .1ex}
\titlespacing*{\subsubsection}
{0pt}{0.75ex plus 0.15ex minus .05ex}{0.75ex plus 0.5ex}

\definecolor{UBCblue}{rgb}{0.0, 0.25, 0.75}

\usepackage{float}
\usepackage[
    colorlinks = true,
    linkcolor  = UBCblue,
    citecolor  = UBCblue,
    urlcolor   = UBCblue,
    hyperfootnotes = false,
]{hyperref}

\newcommand{\sep}[0]{ \; |\; }

\usepackage[
    justification = centering,
    format        = plain,
    labelfont     = {bf,it},
    textfont      = it,
    belowskip     = -5pt,
    labelsep      = period
]{caption}

\usepackage[
    backend      = biber,
    style        = philosophy-classic,
    maxcitenames = 3,
    scauthorsbib = true,
    dashed       = true,
    doi          = false,
    isbn         = false,
    url          = true,
    eprint       = false,
    backref      = true,
]{biblatex}

\renewbibmacro{in:}{}
\DeclareDelimFormat[cbx@textcite]{nameyeardelim}{\addspace}
\newcommand\citenop[1]{\citeauthor{#1}, \citeyear{#1}}
\let\cite\textcite

\DefineBibliographyStrings{english}{
    backrefpage = {\lowercase{p}\adddot },
    backrefpages = {\lowercase{p}p\adddot }
}

\AtEveryBibitem{
    \ifentrytype{article}{\clearfield{url}}{}
    \ifentrytype{book}{\clearfield{url}}{}
    \clearlist{address}
    \clearfield{eprint}
    \clearfield{note}
    \clearfield{language}
    \clearfield{isbn}
    \clearfield{issn}
    \clearlist{location}
    \clearfield{series}
    \clearfield{urldate}
    \clearfield{urlyear}
    \clearfield{urlmonth}
    \clearfield{urlday}
    \ifentrytype{book}{}{
        \clearlist{publisher}
        \clearname{editor}
    }
}

\usepackage{tikz}
\usetikzlibrary{patterns,intersections,arrows,decorations.pathreplacing,decorations.pathmorphing,positioning,arrows.meta,calc,decorations.markings,shapes.misc,matrix,shapes,fit,tikzmark}
\usetikzlibrary{calc,shadings}
\usepackage{pgfplots}
\pgfplotsset{compat = 1.16}

\pgfdeclarepatternformonly{north east lines wide}%
{\pgfqpoint{-1pt}{-1pt}}%
{\pgfqpoint{10pt}{10pt}}%
{\pgfqpoint{9pt}{9pt}}%
{
    \pgfsetlinewidth{0.9pt}
    \pgfpathmoveto{\pgfqpoint{0pt}{0pt}}
    \pgfpathlineto{\pgfqpoint{9.1pt}{9.1pt}}
    \pgfusepath{stroke}
}

\usepackage[
    margin = 1.1in,
    top    = 1.25in,
    bottom = 1.25in,
]{geometry}
\usepackage[flushmargin]{footmisc}

\titlespacing*{\paragraph}{0pt}{0.5ex plus .1ex minus .1ex}{1em}

\usepackage[strict]{changepage}

\makeatletter
\def\th@plain{%
    \thm@notefont{}
    \itshape 
}
\def\th@definition{%
    \thm@notefont{}
    \normalfont 
}
\makeatother

\newtheorem{theorem}{\sc Theorem}
\newtheorem{proposition}{\sc Proposition}

\newtheorem{lemma}{\sc Lemma}
\newtheorem{lemma-app}{\sc Lemma}[section]
\newtheorem{assumption}{\sc Assumption}
\newtheorem{corollary}{\sc Corollary}
\newtheorem{definition}{\sc Definition}
\newtheorem{definition-app}{\sc Definition}[section]

\theoremstyle{definition}

\usepackage{chngcntr}
\usepackage{apptools}
\AtAppendix{\counterwithin{lemma}{section}}

\usepackage{datetime}
\newdateformat{monthyeardate}{%
    \monthname[\THEMONTH] \THEYEAR}

\usepackage{xpatch}

\providecommand{\proofnamefont}{\itshape}
\xpatchcmd{\proof}{\itshape}{\normalfont\proofnamefont}{}{}
\renewcommand{\proofnamefont}{\bfseries}

\usepackage{enumitem}

\setlist[enumerate, 1]{
    itemsep   = 0.05em,
    parsep    = 0.1em,
    partopsep = 0.1em,
    topsep    = 0.1em,
}

\setlist[itemize, 1]{
    itemsep   = 0.05em,
    parsep    = 0.1em,
    partopsep = 0.1em,
    topsep    = 0.1em,
    label     = $\bullet$,
}

\setlist[itemize, 2]{
    itemsep   = 0.15em,
    parsep    = 0.4em,
    partopsep = 0.45em,
    label     = --,
}

\setlist[itemize, 3]{
    itemsep   = 0.1em,
    parsep    = 0.4em,
    partopsep = 0.45em,
    label     = $\triangleright$
}

\newcommand\invisiblesection[1]{%
    \refstepcounter{section}%
    \addcontentsline{toc}{section}{\protect\numberline{\thesection}#1}%
    \sectionmark{#1}}

\def\addlegendimage{\csname pgfplots@addlegendimage\endcsname}

\usepackage{nameref}
\usepackage[capitalise,noabbrev,nameinlink]{cleveref}
\crefname{table}{Table}{tables}
\crefname{equation}{Equation}{equations}
\crefname{lemma-app}{Lemma}{lemmas}
\crefname{claim}{Claim}{claims}
\crefname{problem}{Problem}{problems}
\crefname{example}{Example}{examples}
\crefname{assumption}{Assumption}{Assumptions}
\crefname{definition}{Definition}{Definition}
\crefname{definition-app}{Definition}{Definition}

\newcommand{\floor}[1]{\lfloor #1 \rfloor}
\newcommand{\ceil}[1]{\lceil #1 \rceil}

\usepackage{comment}

%% file: figures/fig-intro.tex
\begin{tikzpicture}[scale = 6.65]


\def\eps{0.00};

\def\vL{-0.2};
\def\aL{-0.4};  
\def\cL{-0.0};  
\def\bL{0.2};  

\def\blevel{0};
\def\rs{-0.08};
\def\labelshift{3pt};

\def\vR{0.4};
\def\bR{0.8};  
\def\cR{0.0};  
\def\aR{-0.4};  


\draw (-1,0) -- (1,0) node[right] {};


\foreach \x in {-1,-0.9,...,1}{
    \draw (\x cm,0.3pt) -- (\x cm,-0.3pt);
}

\foreach \x in {-1,-0.4,-0.2,0,0.2,0.4,0.8,1}{
    \draw (\x cm,0.48pt) -- (\x cm,-0.48pt) node[below] {$\x$};
}


\draw [color = MyBlue, fill = MyBlue] (\vL,\blevel-\rs) circle (.015);
    \node[above, yshift = \labelshift] at ( \vL,\blevel-\rs) {$\textcolor{MyBlue}{L}$};

\draw[MyBlue, line width = 1.10mm] (\aL,\blevel-\rs) -- (\cL,\blevel-\rs);

\draw[MyBlue, line width = 1.10mm, shift={(\aL,\blevel-\rs)}] (0pt,0.49pt) -- (0pt,-0.49pt);

\draw[MyBlue, line width = 1.10mm, shift={(\cL,\blevel-\rs)}] (0pt,0.49pt) -- (0pt,-0.49pt);


\draw[MyBlue, dotted, line width = 1.10mm] (\cL,\blevel-\rs) -- (\bL,\blevel-\rs);

\draw[MyBlue, line width = 1.10mm, shift={(\bL,\blevel-\rs)}] (0pt,0.49pt) -- (0pt,-0.49pt);

\def\rs{-0.04};

\draw [color = MyRed, fill = MyRed] (\vR,\blevel-\rs) circle (.015);
    \node[above, yshift = \labelshift] at ( \vR,\blevel-\rs) {$\textcolor{MyRed}{R}$};

\draw[MyRed, line width = 1.10mm] (\cR,\blevel-\rs) -- (\bR,\blevel-\rs);

\draw[MyRed, line width = 1.10mm, shift={(\bR,\blevel-\rs)}] (0pt,0.49pt) -- (0pt,-0.49pt);

\draw[MyRed, line width = 1.10mm, shift={(\cR,\blevel-\rs)}] (0pt,0.49pt) -- (0pt,-0.49pt);

\draw[MyRed, dotted, line width = 1.10mm] (\aR,\blevel-\rs) -- (\cR,\blevel-\rs);

\draw[MyRed, line width = 1.10mm, shift={(\aR,\blevel-\rs)}] (0pt,0.49pt) -- (0pt,-0.49pt);


\draw[black, line width = 1.50mm] (\aR,\blevel) -- (\bL,\blevel);

\draw[black, line width = 1.10mm, shift={(\aR,\blevel)}] (0pt,0.49pt) -- (0pt,-0.49pt);

\draw[black, line width = 1.10mm, shift={(\bL,\blevel)}] (0pt,0.49pt) -- (0pt,-0.49pt);
    
\end{tikzpicture}

%% file: figures/fig-intro-2.tex
\begin{tikzpicture}[scale = 6.65]


\def\eps{0.00};

\def\vL{-0.2};
\def\aL{-0.4};  
\def\cL{-0.0};  
\def\bL{0.2};  

\def\aLAS{-0.5};
\def\bLAS{0.1791};

\def\blevel{0};
\def\rs{-0.08};
\def\labelshift{3pt};

\def\vR{0.5};
\def\bR{1.0};  
\def\cR{0.0};  
\def\aR{-0.5};  


\draw (-1,0) -- (1,0) node[right] {};


\foreach \x in {-1,-0.9,...,1}{
    \draw (\x cm,0.3pt) -- (\x cm,-0.3pt);
}

\foreach \x in {-1,-0.4,-0.2,0,0.2,0.5,1}{
    \draw (\x cm,0.5pt) -- (\x cm,-0.5pt) node[below] {$\x$};
}

\draw (-0.5 cm,0.5pt) -- (-0.5 cm,-0.5pt) node[below,xshift=-7pt] {$-0.5$};


\draw [color = MyBlue, fill = MyBlue] (\vL,\blevel-\rs) circle (.015);
    \node[above, yshift = \labelshift] at ( \vL,\blevel-\rs) {$\textcolor{MyBlue}{L}$};

\draw[MyBlue, line width = 1.10mm] (\aL,\blevel-\rs) -- (\cL,\blevel-\rs);

\draw[MyBlue, line width = 1.10mm, shift={(\aL,\blevel-\rs)}] (0pt,0.49pt) -- (0pt,-0.49pt);

\draw[MyBlue, line width = 1.10mm, shift={(\cL,\blevel-\rs)}] (0pt,0.49pt) -- (0pt,-0.49pt);


\draw[MyBlue, dotted, line width = 1.10mm] (\aLAS,\blevel-\rs) -- (\bLAS,\blevel-\rs);

\draw[MyBlue, line width = 1.10mm, shift={(\aLAS,\blevel-\rs)}] (0pt,0.49pt) -- (0pt,-0.49pt);

\draw[MyBlue, line width = 1.10mm, shift={(\bLAS,\blevel-\rs)}] (0pt,0.49pt) -- (0pt,-0.49pt);

\def\rs{-0.04};

\draw [color = MyRed, fill = MyRed] (\vR,\blevel-\rs) circle (.015);
    \node[above, yshift = \labelshift] at ( \vR,\blevel-\rs) {$\textcolor{MyRed}{R'}$};

\draw[MyRed, line width = 1.10mm] (\cR,\blevel-\rs) -- (\bR,\blevel-\rs);

\draw[MyRed, line width = 1.10mm, shift={(\bR,\blevel-\rs)}] (0pt,0.49pt) -- (0pt,-0.49pt);

\draw[MyRed, line width = 1.10mm, shift={(\cR,\blevel-\rs)}] (0pt,0.49pt) -- (0pt,-0.49pt);

\draw[MyRed, dotted, line width = 1.10mm] (\aR,\blevel-\rs) -- (\cR,\blevel-\rs);

\draw[MyRed, line width = 1.10mm, shift={(\aR,\blevel-\rs)}] (0pt,0.49pt) -- (0pt,-0.49pt);


\draw[black, line width = 1.50mm] (\aR,\blevel) -- (\bLAS,\blevel);

\draw[black, line width = 1.10mm, shift={(\aR,\blevel)}] (0pt,0.49pt) -- (0pt,-0.49pt);

\draw[black, line width = 1.10mm, shift={(\bLAS,\blevel)}] (0pt,0.49pt) -- (0pt,-0.49pt);
    
\end{tikzpicture}

%% file: figures/fig-Ls-preferences.tex
\begin{tikzpicture}[scale = 1.15]

    \def\r{1.50};

    \def\labelshift{5pt}

    \def\eps{0.50};

    \def\vL{-2};
    \def\aL{-4.00};  
    \def\cL{-0.00};  

    \def\X{0}; 
    \def\Y{-1.40};


    \draw[line width = 1.00pt] (5,0) -- (-5,0);
    \draw [dashed] (-5,\X) -- (-5,\Y-0.03) node[below] {$\textcolor{black}{-1}$};
    \draw[black, line width=1pt, shift={(-5,0)}] (0pt,3pt) -- (0pt,-3pt);
    \draw [dashed] ( 5,\X) -- ( 5,\Y-0.03) node[below] {$\textcolor{black}{ 1}$};
    \draw[black, line width=1pt, shift={(5,0)}] (0pt,3pt) -- (0pt,-3pt);

    \draw [dashed] ( 0,\X) -- ( 0,\Y-0.00) node[below] {$\textcolor{black}{0}$};
    \draw[black, line width=2.4pt, shift={(0,0)}] (0pt,3pt) -- (0pt,-3pt);


    \node[above, yshift=\labelshift] at ( \vL,0) {$\textcolor{MyBlue}{v}$};
    \draw [color = MyBlue, fill = MyBlue] (\vL,0) circle (.075);

    \draw[MyBlue, line width = 0.75mm] plot[smooth,domain=\aL:\cL] (\x, { -1/(2*abs(\vL)) * ((\vL-\x)*(\vL-\x) - abs(\vL)*abs(\vL)) } );

    \draw[MyBlue, dashed, line width = 0.75mm]   plot[smooth,domain=-5:\aL] (\x, { -1/(2*abs(\vL)) * ((\vL-\x)*(\vL-\x) - abs(\vL)*abs(\vL)) } );

    \draw[MyBlue, dashed, line width = 0.75mm] plot[smooth,domain=\cL:1.15] (\x, { -1/(2*abs(\vL)) * ((\vL-\x)*(\vL-\x) - abs(\vL)*abs(\vL)) } );

    \node[below] at (1.50,\Y) {$\lblue{\alpha_{v}(x)}$};

    \draw [dashed] (\cL,\X) -- (\cL,\Y);

    \draw[MyBlue, line width=2.75pt, shift={(\cL,0)}] (0pt,3pt) -- (0pt,-3pt) node[above] {};

    \node[above, xshift=-7pt,yshift=2.5pt] (eps) at (\aL,0) {$\lblue{\floor{A_v}}$};
    \node[above, xshift=7pt,yshift=2.5pt] (eps) at (0,0) {$\lblue{\ceil{A_v}}$};

    \draw[MyBlue, line width=2.75pt, shift={(\aL,0)}] (0pt,3pt) -- (0pt,-3pt) node[above] {};

    \draw [thick, MyBlue, decorate, decoration = {mirror,brace, amplitude = 5.0pt}] (\aL, -0.25) -- (\cL, -0.25) node[MyBlue, midway, below, yshift = -0.2cm] {$\lblue{v\text{'s approval set } A_v}$};

\end{tikzpicture}

%% file: figures/fig-baseline-public.tex
\begin{tikzpicture}[scale = 1.25]

  \def\eps{0.00};

  \def\vL{-1.5};
  \def\aL{-3.00};  
  \def\cL{-0.00};  

  \def\aaL{-4.0981};
  \def\bbL{1.0981};

  \def\Y{-1.5000};



  \def\labelshift{4pt}

  \draw (-5,0) -- (5,0);
  \draw[shift={(-5,0)}] (0pt,4pt) -- (0pt,-4pt) node[below] {$\textcolor{black}{-1}$};
  \draw[shift={(5,0)}] (0pt,4pt) -- (0pt,-4pt) node[below] {$\textcolor{black}{1}$};

  \fill [MyBlue!25, domain=\aL:\vL, variable=\x]
  (\vL, 0)
  -- plot (\x, { -1/(2*abs(\vL)) * ((\vL-\x)*(\vL-\x) - abs(\vL)*abs(\vL)) } )
  -- (\vL, 0)
  -- cycle;

  \fill [MyBlue!25, domain=\vL:\cL, variable=\x]
  (\vL, 0)
  -- plot (\x, { -1/(2*abs(\vL)) * ((\vL-\x)*(\vL-\x) - abs(\vL)*abs(\vL)) } )
  -- (\cL, 0)
  -- cycle;


  \fill [color=MyBlue, pattern=north east lines wide, pattern color=MyBlue!75, domain=\aL:\aaL, variable=\x]
  (\vL, 0)
  -- plot (\x, { -1/(2*abs(\vL)) * ((\vL-\x)*(\vL-\x) - abs(\vL)*abs(\vL)) } )
  -- (\aaL, 0)
  -- cycle;

  \draw [dashed] (\aaL,0) -- (\aaL,\Y);

  \fill [color=MyBlue, pattern=north east lines wide, pattern color=MyBlue!75, domain=\cL:\bbL, variable=\x]
  (\vL, 0)
  -- plot (\x, { -1/(2*abs(\vL)) * ((\vL-\x)*(\vL-\x) - abs(\vL)*abs(\vL)) } )
  -- (\bbL, 0)
  -- cycle;

  \draw [dashed] (\bbL,0) -- (\bbL,\Y);

  \draw[dashed] (-5,\Y) -- (3.40,\Y) node[right] {$\textcolor{black}{c_L^*(-1,1)}$};


  \draw[black, line width = 0.35mm] plot[smooth,domain=\aaL:\bbL] (\x, { -1/(2*abs(\vL)) * ((\vL-\x)*(\vL-\x) - abs(\vL)*abs(\vL)) } );

  \node[below] at (-3.15,-0.55) {$\textcolor{black}{\alpha_{L}(x)}$};

  \node[below, yshift = -\labelshift] (v_L) at (\vL,0) {$\lblue{L}$};
  \draw [color = MyBlue, fill = MyBlue] (\vL,0) circle (.090);


  \draw[MyBlue, line width = 1.1mm] (\aL,0) -- (\cL,0);
  \draw[MyBlue, line width=2.4pt, shift={(\aL,0)}] (0pt,4pt) -- (0pt,-4pt) node[above] {};
  \node[above, xshift = -7.5pt, yshift = \labelshift] (a_L) at (\aL,0) {$\lblue{\floor{A_L}}$};
  \draw[MyBlue, line width=2.4pt, shift={(\cL,0)}] (0pt,4pt) -- (0pt,-4pt) node[above] {};
  \node[above, yshift = \labelshift] (c_L) at (\cL,0) {$\textcolor{black}{0}$};


  \draw[MyBlue, line width = 1.25mm, dashed] (\cL,0) -- (\bbL,0);
  \draw[MyBlue, line width=2.4pt, shift={(\bbL,0)}] (0pt,4pt) -- (0pt,-4pt);
  \node[above, yshift = \labelshift] (b_L) at (\bbL,0) { };

  \draw[MyBlue, line width = 1.25mm, dashed] (\aaL,0) -- (\aL,0);
  \draw[MyBlue, line width=2.4pt, shift={(\aaL,0)}] (0pt,4pt) -- (0pt,-4pt);
  \node[above, yshift = \labelshift] (b_L) at (\aaL,0) { };


  \draw [thick, MyBlue, decorate, decoration = {brace, amplitude = 5.0pt}] (\aaL, 1.05) -- (\bbL, 1.05) node[MyBlue, midway, above, yshift = 0.2cm] {$\lblue{ I_L(-1,1) }$};

\end{tikzpicture}

%% file: figures/fig-largest-interval-AS-L.tex
\begin{tikzpicture}[scale = 1.25]

\def\eps{0.00};

\def\vL{-1};
\def\aL{-2.00};  
\def\cL{-0.00};  

\def\aaL{-2.00};
\def\bbL{1.0};

\def\Y{-1.5};

\draw[draw=none] (0,1) -- (0,-2);



\def\labelshift{4pt}

\draw (-3,0) -- (2,0);

\draw[MyBlue, line width = 1.25mm, dashed] (\cL,0) -- (\bbL,0);
\draw[MyBlue, line width=2.4pt, shift={(\bbL,0)}] (0pt,4pt) -- (0pt,-4pt);
\node[above, yshift = \labelshift] (b_L) at (\bbL,0) {$\lblue{b_L}$};

\fill [MyBlue!25, domain=\aL:\vL, variable=\x]
  (\vL, 0)
  -- plot (\x, { -1/(2*abs(\vL)) * ((\vL-\x)*(\vL-\x) - abs(\vL)*abs(\vL)) } )
  -- (\vL, 0)
  -- cycle;

\fill [MyBlue!25, domain=\vL:\cL, variable=\x]
  (\vL, 0)
  -- plot (\x, { -1/(2*abs(\vL)) * ((\vL-\x)*(\vL-\x) - abs(\vL)*abs(\vL)) } )
  -- (\cL, 0)
  -- cycle;


\fill [color=MyBlue, pattern=north east lines wide, pattern color=MyBlue!75, domain=\cL:\bbL, variable=\x]
  (\vL, 0)
  -- plot (\x, { -1/(2*abs(\vL)) * ((\vL-\x)*(\vL-\x) - abs(\vL)*abs(\vL)) } )
  -- (\bbL, 0)
  -- cycle;

\draw [dashed] (\bbL,0) -- (\bbL,\Y);


\node[below, yshift = -\labelshift] (v_L) at (\vL,0) {$\lblue{L}$};
\draw [color = MyBlue, fill = MyBlue] (\vL,0) circle (.090);


\draw[black, line width = 0.35mm] plot[smooth,domain=\aaL:\bbL] (\x, { -1/(2*abs(\vL)) * ((\vL-\x)*(\vL-\x) - abs(\vL)*abs(\vL)) } );

  \node[below] at (0.25,-0.75) {$\textcolor{black}{\alpha_{L}(x)}$};


\draw[MyBlue, line width = 1.1mm] (\aL,0) -- (\cL,0);
\draw[MyBlue, line width=2.4pt, shift={(\aL,0)}] (0pt,4pt) -- (0pt,-4pt) node[above] {};
\node[above, xshift = -7.5pt, yshift = \labelshift] (a_L) at (\aL,0) {$\lblue{\floor{A_L}}$};
\draw[MyBlue, line width=2.4pt, shift={(\cL,0)}] (0pt,4pt) -- (0pt,-4pt) node[above] {};
\node[above, yshift = \labelshift] (c_L) at (\cL,0) {$\textcolor{black}{0}$};




\end{tikzpicture}

%% file: figures/fig-largest-interval-AS-R.tex
\begin{tikzpicture}[scale = 1.25]

\def\eps{0.00};

\def\vR{1.25};
\def\aRAS{-1.25};
\def\aR{0.00};
\def\bR{2.50};

\def\Y{-1.75};

\draw[draw=none] (0,1) -- (0,-2);



\def\labelshift{4pt}

\draw (-2,0) -- (3,0);

\fill [MyRed!25, domain=\aR:\vR, variable=\x]
  (\vR, 0)
  -- plot (\x, { -1/(2*abs(\vR)) * ((\vR-\x)*(\vR-\x) - abs(\vR)*abs(\vR)) } )
  -- (\vR, 0)
  -- cycle;

\fill [MyRed!25, domain=\vR:\bR, variable=\x]
  (\vR, 0)
  -- plot (\x, { -1/(2*abs(\vR)) * ((\vR-\x)*(\vR-\x) - abs(\vR)*abs(\vR)) } )
  -- (\aR, 0)
  -- cycle;


\fill [color=MyRed, pattern=north east lines wide, pattern color=MyRed!75, domain=0:\aRAS, variable=\x]
  (\vR, 0)
  -- plot (\x, { -1/(2*abs(\vR)) * ((\vR-\x)*(\vR-\x) - abs(\vR)*abs(\vR)) } )
  -- (\aRAS, 0)
  -- cycle;

\draw [dashed] (\aRAS,0) -- (\aRAS,\Y);


\node[below, yshift = -\labelshift] (v_L) at (\vR,0) {$\lred{R}$};
\draw [color = MyRed, fill = MyRed] (\vR,0) circle (.090);


\draw[black, line width = 0.35mm] plot[smooth,domain=\aRAS:\bR] (\x, { -1/(2*abs(\vR)) * ((\vR-\x)*(\vR-\x) - abs(\vR)*abs(\vR)) } );

  \node[below] at (-0.25,-0.75) {$\textcolor{black}{\alpha_{R}(x)}$};

\draw[MyRed, line width = 1.25mm, dashed] (\aRAS,0) -- (0,0);
\draw[MyRed, line width=2.4pt, shift={(\aRAS,0)}] (0pt,4pt) -- (0pt,-4pt);
\node[above, yshift = \labelshift] at (\aRAS,0) {$\lred{a_R}$};


\draw[MyRed, line width = 1.1mm] (\aR,0) -- (\bR,0);
\draw[MyRed, line width=2.4pt, shift={(\aR,0)}] (0pt,4pt) -- (0pt,-4pt) node[above] {};
\node[above, xshift = 7.5pt, yshift = \labelshift] at (\bR,0) {$\lred{\ceil{A_R}}$};
\draw[MyRed, line width=2.4pt, shift={(\bR,0)}] (0pt,4pt) -- (0pt,-4pt) node[above] {};
\node[above, yshift = \labelshift] at (\aR,0) {$\textcolor{black}{0}$};



\end{tikzpicture}

%% file: figures/fig-baseline-TA-interior.tex
\begin{tikzpicture}[scale = 1.25]


\def\eps{0.00};

\def\vL{-0.95};
\def\aL{-1.90};  
\def\cL{-0.00};  
\def\bL{0.95};  

\def\blevel{0};
\def\rs{0.25};
\def\labelshift{5pt};


\draw (-5,0) -- (5,0) node[right] {};
\draw[shift={(-5,0)}] (0pt,3pt) -- (0pt,-3pt) node[below] {$\textcolor{black}{-1}$};
\draw[shift={(5,0)}] (0pt,3pt) -- (0pt,-3pt) node[below] {$\textcolor{black}{1}$};
	\node[below, yshift = -\labelshift] at (0,\blevel-\rs) {$\textcolor{black}{0}$};
\draw[black, line width=2.4pt, shift={(0,0)}] (0pt,3pt) -- (0pt,-3pt) node[above] {};

\def\blevel{0};

\def\vR{1.35};
\def\bR{2.70};  
\def\cR{0.00};  
\def\aR{-1.35};  


\draw [line width=2.4pt] (\aR, 1.25) -- (\bL, 1.25) node[MyBlue, midway, above, yshift = 0.05cm] {$\textcolor{black}{\text{policy outcomes approved by } \{L,R\} }$};

    \draw[line width=2.4pt, shift={(\aR,1.25)}] (0pt,3pt) -- (0pt,-3pt);
    \draw[line width=2.4pt, shift={(\bL,1.25)}] (0pt,3pt) -- (0pt,-3pt);

    \draw [loosely dotted] (\aR, 0) -- (\aR, 1.25);
    \draw [loosely dotted] (\bL, 0) -- (\bL, 1.25);


\node[above, yshift = \labelshift] (v_L) at (\vL,\blevel+\rs) {$\lblue{L}$};
\draw [color = MyBlue, fill = MyBlue] (\vL,\blevel+\rs) circle (.075);

\draw[MyBlue, line width = 1.00mm] (\aL,\blevel+\rs) -- (\cL,\blevel+\rs);

\draw[MyBlue, line width=2.4pt, shift={(\aL,\blevel+\rs)}] (0pt,3pt) -- (0pt,-3pt) node[above] {};

\node[above, yshift = \labelshift] at (\aL,\blevel+0.20) {$\lblue{ \floor{A_L} }$};

\draw[MyBlue, line width=2.4pt, shift={(\cL,\blevel+\rs)}] (0pt,3pt) -- (0pt,-3pt) node[above] {};



\draw[MyBlue, line width = 1.00mm, dotted] (\cL,\blevel+\rs) -- (\bL,\blevel+\rs);

\draw[MyBlue, line width=2.4pt, shift={(\bL,\blevel+\rs)}] (0pt,3pt) -- (0pt,-3pt) node[above] {};
	
\node[above, yshift = \labelshift] at (\bL,\blevel+\rs) {$\lblue{b_L}$};



\node[below, yshift = -\labelshift] (v_R) at (\vR,\blevel-\rs) {$\lred{R}$};
\draw [color = MyRed, fill = MyRed] (\vR,\blevel-\rs) circle (.075);

\draw[MyRed, line width = 1.00mm] (\cR,\blevel-\rs) -- (\bR,\blevel-\rs);

\draw[MyRed, line width=2.4pt, shift={(\cR,\blevel-\rs)}] (0pt,3pt) -- (0pt,-3pt) node[below] {};


\draw[MyRed, line width=2.4pt, shift={(\bR,\blevel-\rs)}] (0pt,3pt) -- (0pt,-3pt) node[below] {};

\node[below, yshift = -\labelshift] at (\bR,\blevel-\rs) {$\lred{ \ceil{A_R} }$};


\draw[MyRed, line width = 1.00mm, dotted] (\aR,\blevel-\rs) -- (\cR,\blevel-\rs);

\draw[MyRed, line width=2.4pt, shift={(\aR,\blevel-\rs)}] (0pt,3pt) -- (0pt,-3pt) node[below] {};

\node[below, yshift = -\labelshift] at (\aR,\blevel-\rs) {$\lred{a_R}$};


\end{tikzpicture}

%% file: figures/fig-baseline-TA-exterior.tex
\begin{tikzpicture}[scale = 1.25]


\def\eps{0.00};

\def\vL{-0.95};
\def\aL{-1.90};  
\def\cL{-0.00};  
\def\bL{0.95};  

\def\blevel{0};
\def\rs{0.25};
\def\labelshift{5pt};


\draw (-5,0) -- (5,0) node[right] {};
\draw[shift={(-5,0)}] (0pt,3pt) -- (0pt,-3pt) node[below] {$\textcolor{black}{-1}$};
\draw[shift={(5,0)}] (0pt,3pt) -- (0pt,-3pt) node[below] {$\textcolor{black}{1}$};
	\node[below, yshift = -\labelshift] at (0,\blevel-\rs) {$\textcolor{black}{0}$};
\draw[black, line width=2.4pt, shift={(0,0)}] (0pt,3pt) -- (0pt,-3pt) node[above] {};

\def\blevel{0};

\def\vR{2.40};
\def\bR{4.80};  
\def\cR{0.00};  
\def\aR{-2.40};  


\node[above, yshift = \labelshift] (v_L) at (\vL,\blevel+\rs) {$\lblue{L}$};
\draw [color = MyBlue, fill = MyBlue] (\vL,\blevel+\rs) circle (.075);

\draw[MyBlue, line width = 1.00mm] (\aL,\blevel+\rs) -- (\cL,\blevel+\rs);

\draw[MyBlue, line width=2.4pt, shift={(\aL,\blevel+\rs)}] (0pt,3pt) -- (0pt,-3pt) node[above] {};

\node[above, yshift = \labelshift] at (\aL,\blevel+0.20) {$\lblue{ \floor{A_L} }$};

\draw[MyBlue, line width=2.4pt, shift={(\cL,\blevel+\rs)}] (0pt,3pt) -- (0pt,-3pt) node[above] {};




	



\node[below, yshift = -\labelshift] (v_R) at (\vR,\blevel-\rs) {$\lred{R}$};
\draw [color = MyRed, fill = MyRed] (\vR,\blevel-\rs) circle (.075);

\draw[MyRed, line width = 1.00mm] (\cR,\blevel-\rs) -- (\bR,\blevel-\rs);

\draw[MyRed, line width=2.4pt, shift={(\cR,\blevel-\rs)}] (0pt,3pt) -- (0pt,-3pt) node[below] {};


\draw[MyRed, line width=2.4pt, shift={(\bR,\blevel-\rs)}] (0pt,3pt) -- (0pt,-3pt) node[below] {};

\node[below, yshift = -\labelshift] at (\bR,\blevel-\rs) {$\lred{ \ceil{A_R} }$};


\draw[MyRed, line width = 1.00mm, dotted] (\aR,\blevel-\rs) -- (\cR,\blevel-\rs);

\draw[MyRed, line width=2.4pt, shift={(\aR,\blevel-\rs)}] (0pt,3pt) -- (0pt,-3pt) node[below] {};

\node[below, yshift = -\labelshift] at (\aR,\blevel-\rs) {$\lred{a_R}$};



\def\aLS{-2.40};
\def\bLS{0.8383};

\draw[MyBlue, line width = 1.00mm, dotted] (\aLS,\blevel+\rs) -- (\bLS,\blevel+\rs);

\draw[MyBlue, line width=2.4pt, shift={(\aLS,\blevel+\rs)}] (0pt,3pt) -- (0pt,-3pt) node[above] {};
\draw[MyBlue, line width=2.4pt, shift={(\bLS,\blevel+\rs)}] (0pt,3pt) -- (0pt,-3pt) node[above] {};


\draw [line width=2.4pt] (\aR, 1.25) -- (\bLS, 1.25) node[MyBlue, midway, above, yshift = 0.05cm] {$\textcolor{black}{\text{policy outcomes approved by } \{L,R\} }$};

    \draw[line width=2.4pt, shift={(\aR,1.25)}] (0pt,3pt) -- (0pt,-3pt);
    \draw[line width=2.4pt, shift={(\bLS,1.25)}] (0pt,3pt) -- (0pt,-3pt);

    \draw [loosely dotted] (\aR, 0) -- (\aR, 1.25);
    \draw [loosely dotted] (\bLS, 0) -- (\bLS, 1.25);

\end{tikzpicture}

%% file: figures/fig-R-more-extreme.tex
\begin{tikzpicture}[scale = 0.95]


  \def\Ydown{-4.5937};

  \def\vR{1.35};
  \def\aRAS{-1.35};
  \def\aR{0.00};
  \def\bR{2.70};

  \def\rs{0.085};
  \def\labelshift{3pt};
  \def\blevel{0};



  \draw[MyRed, line width = 0.50mm] plot[smooth,domain=\aR:\bR] (\x, { -0.50* ((\vR-\x)*(\vR-\x) - abs(\vR)*abs(\vR)) } );


  \draw[MyRed, dashed, line width = 0.50mm]   plot[smooth,domain=-1.35:\aR] (\x, { -0.50* ((\vR-\x)*(\vR-\x) - abs(\vR)*abs(\vR)) } );

  \draw[MyRed, dashed, line width = 0.50mm] plot[smooth,domain=\bR:3.35] (\x, { -0.50* ((\vR-\x)*(\vR-\x) - abs(\vR)*abs(\vR)) } );

  \node[below] at (3.50,-1.10) {$\lred{\alpha_{v}(x)}$};


  \draw (-5,0) -- (5,0);
  \draw[shift={(-5,0)}] (0pt,2.5pt) -- (0pt,-2.5pt) node[below] {$\textcolor{black}{-1}$};
  \draw[shift={(5,0)}] (0pt,2.5pt) -- (0pt,-2.5pt) node[below] {$\textcolor{black}{1}$};

  \node[below, yshift = -\labelshift] (v_v) at (\vR,\blevel-\rs) {$\lred{v}$};
  \draw [color = MyRed, fill = MyRed] (\vR,\blevel-\rs) circle (.075);

  \draw[MyRed, line width = 1.00mm] (\aR,\blevel-\rs) -- (\bR,\blevel-\rs);

  \draw[MyRed, line width=2.4pt, shift={(\bR,\blevel-\rs)}] (0pt,2.5pt) -- (0pt,-2.5pt) node[below] {};
  \node[below, yshift = -\labelshift, xshift = -6pt] at (\bR,\blevel-\rs) {$\lred{ \ceil{A_{v}} }$};


  \draw[MyRed, line width = 1.00mm, dotted] (\aRAS,\blevel-\rs) -- (\aR,\blevel-\rs);

  \draw[MyRed, line width=2.4pt, shift={(\aRAS,\blevel-\rs)}] (0pt,2.5pt) -- (0pt,-2.5pt) node[below] {};

  \node[below, yshift = -\labelshift] at (\aRAS,\blevel-\rs) {$\lred{a_{v}}$};


  \def\labelshift{-3pt}


  \def\vR{1.75};
  \def\aRAS{-1.75};
  \def\aR{0.00};
  \def\bR{3.50};

  \def\rs{-0.085};
  \def\blevel{0};





  \draw[MyPurple, line width = 0.50mm] plot[smooth,domain=\aR:\bR] (\x, { -0.50* ((\vR-\x)*(\vR-\x) - abs(\vR)*abs(\vR)) } );


  \draw[MyPurple, dashed, line width = 0.50mm]   plot[smooth,domain=-1.175:\aR] (\x, { -0.50* ((\vR-\x)*(\vR-\x) - abs(\vR)*abs(\vR)) } );

  \draw[MyPurple, dashed, line width = 0.50mm] plot[smooth,domain=\bR:4.5] (\x, { -0.50* ((\vR-\x)*(\vR-\x) - abs(\vR)*abs(\vR)) } );

  \node[below] at (4.5,-2.10) {$\lpurple{\alpha_{w}(x)}$};

  \node[above, yshift = -\labelshift] (v_v) at (\vR,\blevel-\rs) {$\lpurple{w}$};
  \draw [color = MyPurple, fill = MyPurple] (\vR,\blevel-\rs) circle (.075);

  \draw[MyPurple, line width = 1.00mm] (\aR,\blevel-\rs) -- (\bR,\blevel-\rs);

  \draw[MyPurple, line width=2.4pt, shift={(\bR,\blevel-\rs)}] (0pt,2.5pt) -- (0pt,-2.5pt) node[below] {};
  \node[above, yshift = -\labelshift, xshift = 8pt] at (\bR,\blevel-\rs) {$\lpurple{ \ceil{A_{w}} }$};


  \draw[MyPurple, line width = 1.00mm, dotted] (\aRAS,\blevel-\rs) -- (\aR,\blevel-\rs);

  \draw[MyPurple, line width=2.4pt, shift={(\aRAS,\blevel-\rs)}] (0pt,2.5pt) -- (0pt,-2.5pt) node[below] {};

  \node[above, yshift = -\labelshift] at (\aRAS,\blevel-\rs) {$\lpurple{a_{w}}$};

  \draw[black, line width=2.4pt, shift={(0,0)}] (0pt,5pt) -- (0pt,-5pt) node[below] {$\textcolor{black}{0}$};

\end{tikzpicture}

%% file: figures/fig-comp_statics.tex
\begin{tikzpicture}[scale = 0.95]


\def\steplevel{1.50};

\def\vL{-0.75};
\def\aL{-1.50};  
\def\cL{-0.00};  
\def\bL{0.75};  

\def\blevel{0};
\def\rs{0.20};
\def\labelshift{2pt};


\draw (-5,\blevel) -- (5,\blevel) node[right] {};
\draw[shift={(-5,\blevel)}] (0pt,2.5pt) -- (0pt,-2.5pt) node[below] {$\textcolor{black}{-1}$};
\draw[shift={(5,\blevel)}] (0pt,2.5pt) -- (0pt,-2.5pt) node[below] {$\textcolor{black}{1}$};
\draw[black, line width = 1.10mm, shift={(0,\blevel)}] (0pt,2.5pt) -- (0pt,-2.5pt) node[above] {};

\def\blevel{0};

\def\vR{0.75};
\def\bR{1.50};  
\def\cR{0.00};  
\def\aR{-0.75};  


\node[above, yshift = \labelshift] (v_L) at (\vL,\blevel+\rs) {$\lblue{L}$};
\draw [color = MyBlue, fill = MyBlue] (\vL,\blevel+\rs) circle (.075);

\draw[MyBlue, line width = 1.10mm] (\aL,\blevel+\rs) -- (\cL,\blevel+\rs);
\draw[MyBlue, line width = 1.10mm, shift={(\aL,\blevel+\rs)}] (0pt,2.5pt) -- (0pt,-2.5pt) node[above] {};
\draw[MyBlue, line width = 1.10mm, shift={(\cL,\blevel+\rs)}] (0pt,2.5pt) -- (0pt,-2.5pt) node[above] {};

\draw[MyBlue, line width = 1.10mm, dotted] (\cL,\blevel+\rs) -- (\bL,\blevel+\rs);
\draw[MyBlue, line width = 1.10mm, shift={(\bL,\blevel+\rs)}] (0pt,2.5pt) -- (0pt,-2.5pt) node[above] {};


\node[below, yshift = -\labelshift] (v_R) at (\vR,\blevel-\rs) {$\lred{R_1}$};
\draw [color = MyRed, fill = MyRed] (\vR,\blevel-\rs) circle (.075);

\draw[MyRed, line width = 1.10mm] (\cR,\blevel-\rs) -- (\bR,\blevel-\rs);
\draw[MyRed, line width = 1.10mm, shift={(\cR,\blevel-\rs)}] (0pt,2.5pt) -- (0pt,-2.5pt) node[below] {};
\draw[MyRed, line width = 1.10mm, shift={(\bR,\blevel-\rs)}] (0pt,2.5pt) -- (0pt,-2.5pt) node[below] {};

\draw[MyRed, line width = 1.10mm, dotted] (\aR,\blevel-\rs) -- (\cR,\blevel-\rs);
\draw[MyRed, line width = 1.10mm, shift={(\aR,\blevel-\rs)}] (0pt,2.5pt) -- (0pt,-2.5pt) node[below] {};

\draw [line width = 1.10mm] (\aR, \blevel) -- (\bL, \blevel);

\draw[line width = 1.10mm, shift={(\aR,\blevel)}] (0pt,2.5pt) -- (0pt,-2.5pt);
\draw[line width = 1.10mm, shift={(\bL,\blevel)}] (0pt,2.5pt) -- (0pt,-2.5pt);


\def\blevel{-1.10};

\def\vR{1.25};
\def\bR{2.50};  
\def\cR{0.00};  
\def\aR{-1.25};  

\draw[black, line width = 1.10mm, shift={(0,\blevel)}] (0pt,2.5pt) -- (0pt,-2.5pt) node[above] {};


\node[above, yshift = \labelshift] (v_L) at (\vL,\blevel+\rs) {$\lblue{L}$};
\draw [color = MyBlue, fill = MyBlue] (\vL,\blevel+\rs) circle (.075);

\draw[MyBlue, line width = 1.10mm] (\aL,\blevel+\rs) -- (\cL,\blevel+\rs);

\draw[MyBlue, line width = 1.10mm, shift={(\aL,\blevel+\rs)}] (0pt,2.5pt) -- (0pt,-2.5pt) node[above] {};

\draw[MyBlue, line width = 1.10mm, shift={(\cL,\blevel+\rs)}] (0pt,2.5pt) -- (0pt,-2.5pt) node[above] {};


\draw[MyBlue, line width = 1.10mm, dotted] (\aL,\blevel+\rs) -- (\bL,\blevel+\rs);
\draw[MyBlue, line width = 1.10mm, shift={(\bL,\blevel+\rs)}] (0pt,2.5pt) -- (0pt,-2.5pt) node[above] {};


\node[below, yshift = -\labelshift] (v_R) at (\vR,\blevel-\rs) {$\lred{R_2}$};
\draw [color = MyRed, fill = MyRed] (\vR,\blevel-\rs) circle (.075);

\draw[MyRed, line width = 1.10mm] (\cR,\blevel-\rs) -- (\bR,\blevel-\rs);
\draw[MyRed, line width = 1.10mm, shift={(\cR,\blevel-\rs)}] (0pt,2.5pt) -- (0pt,-2.5pt) node[below] {};
\draw[MyRed, line width = 1.10mm, shift={(\bR,\blevel-\rs)}] (0pt,2.5pt) -- (0pt,-2.5pt) node[below] {};

\draw[MyRed, line width = 1.10mm, dotted] (\aR,\blevel-\rs) -- (\cR,\blevel-\rs);
\draw[MyRed, line width = 1.10mm, shift={(\aR,\blevel-\rs)}] (0pt,2.5pt) -- (0pt,-2.5pt) node[below] {};

\draw [line width = 1.10mm] (\aR, \blevel) -- (\bL, \blevel);

\draw[line width = 1.10mm, shift={(\aR,\blevel)}] (0pt,2.5pt) -- (0pt,-2.5pt);
\draw[line width = 1.10mm, shift={(\bL,\blevel)}] (0pt,2.5pt) -- (0pt,-2.5pt);


\def\blevel{-2.20};

\def\vR{1.75};
\def\bR{3.50};  
\def\cR{0.00};  
\def\aR{-1.75};  

\def\WLl{-1.75};
\def\WLr{0.7182};


\draw[black, line width = 1.10mm, shift={(0,\blevel)}] (0pt,2.5pt) -- (0pt,-2.5pt) node[above] {};


\node[above, yshift = \labelshift] (v_L) at (\vL,\blevel+\rs) {$\lblue{L}$};
\draw [color = MyBlue, fill = MyBlue] (\vL,\blevel+\rs) circle (.075);

\draw[MyBlue, line width = 1.10mm] (\aL,\blevel+\rs) -- (\cL,\blevel+\rs);
\draw[MyBlue, line width = 1.10mm, shift={(\aL,\blevel+\rs)}] (0pt,2.5pt) -- (0pt,-2.5pt) node[above] {};
\draw[MyBlue, line width = 1.10mm, shift={(\cL,\blevel+\rs)}] (0pt,2.5pt) -- (0pt,-2.5pt) node[above] {};


\draw[MyBlue, line width = 1.10mm, dotted] (\WLl,\blevel+\rs) -- (\WLr,\blevel+\rs);

\draw[MyBlue, line width = 1.10mm, shift={(\WLl,\blevel+\rs)}] (0pt,2.5pt) -- (0pt,-2.5pt) node[below] {};

\draw[MyBlue, line width = 1.10mm, shift={(\WLr,\blevel+\rs)}] (0pt,2.5pt) -- (0pt,-2.5pt) node[below] {};


\node[below, yshift = -\labelshift] (v_R) at (\vR,\blevel-\rs) {$\lred{R_3}$};
\draw [color = MyRed, fill = MyRed] (\vR,\blevel-\rs) circle (.075);

\draw[MyRed, line width = 1.10mm] (\cR,\blevel-\rs) -- (\bR,\blevel-\rs);
\draw[MyRed, line width = 1.10mm, shift={(\cR,\blevel-\rs)}] (0pt,2.5pt) -- (0pt,-2.5pt) node[below] {};
\draw[MyRed, line width = 1.10mm, shift={(\bR,\blevel-\rs)}] (0pt,2.5pt) -- (0pt,-2.5pt) node[below] {};

\draw[MyRed, line width = 1.10mm, dotted] (\aR,\blevel-\rs) -- (\cR,\blevel-\rs);
\draw[MyRed, line width = 1.10mm, shift={(\aR,\blevel-\rs)}] (0pt,2.5pt) -- (0pt,-2.5pt) node[below] {};

\draw [line width = 1.10mm] (\WLl, \blevel) -- (\WLr, \blevel);

\draw[line width = 1.10mm, shift={(\WLl,\blevel)}] (0pt,2.5pt) -- (0pt,-2.5pt);
\draw[line width = 1.10mm, shift={(\WLr,\blevel)}] (0pt,2.5pt) -- (0pt,-2.5pt);


\def\blevel{-3.30};

\def\vR{2.5};
\def\bR{5.00};  
\def\cR{0.00};  
\def\aR{-2.50};  

\def\WLl{-2.0490};
\def\WLr{0.5490};


\draw[black, line width = 1.10mm, shift={(0,\blevel)}] (0pt,2.5pt) -- (0pt,-2.5pt) node[above] {};


\node[above, yshift = \labelshift] (v_L) at (\vL,\blevel+\rs) {$\lblue{L}$};
\draw [color = MyBlue, fill = MyBlue] (\vL,\blevel+\rs) circle (.075);

\draw[MyBlue, line width = 1.10mm] (\aL,\blevel+\rs) -- (\cL,\blevel+\rs);
\draw[MyBlue, line width = 1.10mm, shift={(\aL,\blevel+\rs)}] (0pt,2.5pt) -- (0pt,-2.5pt) node[above] {};
\draw[MyBlue, line width = 1.10mm, shift={(\cL,\blevel+\rs)}] (0pt,2.5pt) -- (0pt,-2.5pt) node[above] {};


\draw[MyBlue, line width = 1.10mm, dotted] (\WLl,\blevel+\rs) -- (\WLr,\blevel+\rs);

\draw[MyBlue, line width = 1.10mm, shift={(\WLl,\blevel+\rs)}] (0pt,2.5pt) -- (0pt,-2.5pt) node[below] {};

\draw[MyBlue, line width = 1.10mm, shift={(\WLr,\blevel+\rs)}] (0pt,2.5pt) -- (0pt,-2.5pt) node[below] {};


\node[below, yshift = -\labelshift] (v_R) at (\vR,\blevel-\rs) {$\lred{R_4}$};
\draw [color = MyRed, fill = MyRed] (\vR,\blevel-\rs) circle (.075);

\draw[MyRed, line width = 1.10mm] (\cR,\blevel-\rs) -- (\bR,\blevel-\rs);
\draw[MyRed, line width = 1.10mm, shift={(\cR,\blevel-\rs)}] (0pt,2.5pt) -- (0pt,-2.5pt) node[below] {};
\draw[MyRed, line width = 1.10mm, shift={(\bR,\blevel-\rs)}] (0pt,2.5pt) -- (0pt,-2.5pt) node[below] {};

\draw[MyRed, line width = 1.10mm, dotted] (\aR,\blevel-\rs) -- (\cR,\blevel-\rs);
\draw[MyRed, line width = 1.10mm, shift={(\aR,\blevel-\rs)}] (0pt,2.5pt) -- (0pt,-2.5pt) node[below] {};

\draw [line width = 1.10mm] (\WLl, \blevel) -- (\WLr, \blevel);

\draw[line width = 1.10mm, shift={(\WLl,\blevel)}] (0pt,2.5pt) -- (0pt,-2.5pt);
\draw[line width = 1.10mm, shift={(\WLr,\blevel)}] (0pt,2.5pt) -- (0pt,-2.5pt);

\end{tikzpicture}